\documentclass[onefignum,onetabnum]{siamart251216} %

\usepackage{lipsum}
\usepackage{amsfonts}
\usepackage{graphicx}
\usepackage{epstopdf}
\usepackage{mathtools}
\usepackage{amssymb}
\usepackage{bbm}
\usepackage{accents}
\usepackage{algorithm}
\usepackage{algpseudocode}
\usepackage{tabto}
\usepackage{aliascnt}
\usepackage{cleveref}

\ifpdf
  \DeclareGraphicsExtensions{.eps,.pdf,.png,.jpg}
\else
  \DeclareGraphicsExtensions{.eps}
\fi

\AddToHook{cmd/appendix/before}{%
  \crefalias{section}{appendix}
}

\theoremstyle{plain}
\theorembodyfont{\normalfont}
\theoremheaderfont{\normalfont\itshape}

\newaliascnt{remark}{theorem}
\newtheorem{remark}[remark]{Remark}
\aliascntresetthe{remark}

\newaliascnt{condition}{theorem}
\newtheorem{condition}[condition]{Condition}
\aliascntresetthe{condition}

\crefname{hypothesis}{Hypothesis}{Hypotheses}
\crefname{remark}{Remark}{Remarks}
\Crefname{remark}{Remark}{Remarks}
\crefname{condition}{Condition}{Conditions}
\Crefname{condition}{Condition}{Conditions}
\crefname{fact}{Fact}{Facts}
\Crefname{subsection}{Section}{Sections}

\usepackage{enumitem}
\setlist[enumerate]{leftmargin=.5in}
\setlist[itemize]{leftmargin=.5in}

\newcommand{\reverse}[1]{\accentset{\leftarrow}{{#1}}}

\makeatletter
\@mparswitchfalse
\makeatother
\normalmarginpar

\headers{Goal-oriented learning of SDEs}{J.~Zou, H.~C.~Lie, Y.~Marzouk}
\title{Goal-oriented learning of stochastic differential equations using error bounds on path-space observables}

\author{Joanna Zou\thanks{Center for Computational Science \& Engineering, Massachusetts Institute of Technology, Cambridge, MA 
  (\email{jjzou@mit.edu}).}
\and Han Cheng Lie\thanks{Department of Mathematics, University of Potsdam, Potsdam, Germany 
  (\email{han.lie@uni-potsdam.de}).}
\and Youssef Marzouk\thanks{Center for Computational Science \& Engineering, Massachusetts Institute of Technology, Cambridge, MA 
  (\email{ymarz@mit.edu}).}}

\usepackage{amsopn}

\ifpdf
\hypersetup{
  pdftitle={Goal-oriented learning of stochastic differential equations using error bounds on path-space observables},
  pdfauthor={J. Zou, H. C. Lie, Y. Marzouk}
}
\fi

\begin{document}

\maketitle

\begin{abstract}
Stochastic differential equations (SDEs), which serve as the governing equations for dynamical systems in a broad range of applications, can become cost-prohibitive for numerical simulation at scales necessary for quantifying key properties.
Surrogate models of the drift function of an SDE, learned from data of the high-fidelity system, are routinely used to increase the efficiency of simulation and prediction of properties. However, standard choices of loss function for learning the surrogate model fail to provide error guarantees in certain path-dependent observables, such as transition times. This paper introduces an error bound for path-space observables and employs it as a novel variational loss for the goal-oriented learning of the drift function of a SDE. We show the error bound holds for a broad class of observables, including mean first hitting times on unbounded time domains. We derive an analytical gradient of the goal-oriented loss by leveraging the formula for Fr\'echet derivatives of expected path functionals, which remains tractable for implementation in stochastic gradient descent schemes. We demonstrate that surrogate models of overdamped Langevin systems developed via goal-oriented learning achieve improved accuracy in predicting the statistics of a first hitting time observable and robustness to distributional shift in the data.  
\end{abstract}



\section{Introduction}
\label{sec:intro}

Stochastic differential equations (SDEs) are widely used to model stochastic dynamical systems in disciplines such as computational physics, chemistry, and biology, engineering and control systems, and mathematical finance. Due to their random nature, dynamical systems governed by SDEs are generally characterized by their probability laws rather than individual path realizations. A common goal is to quantify statistics of the path expressed as expectations with respect to the law of the solution of the SDE, referred to in this work as \textit{path-space observables}. These quantities relate to time-dependent properties of the system, such as transition probabilities, reaction rates, accumulated costs, or stability metrics, and are useful for downstream tasks in reliability analysis and system design.

Path-space observables are often computationally demanding to quantify, whether by solving a partial differential equation involving the generator of the SDE \cite{Hartmann2012} or using Monte Carlo methods involving multiple long-time path realizations. These costs are compounded for SDEs defined by high-fidelity first-principles models where the drift function is expensive to evaluate. Although surrogate models can substantially reduce the cost of simulation, their utility depends on their ability to reproduce the path-space observable of interest. Standard learning objectives for SDEs based on pointwise accuracy in the drift may not guarantee accuracy in path-space statistics under general conditions, or they require reference data of the observable itself, which are often noisy, inaccurate, or unavailable. 

One example of this challenge is the estimation of the mean first transition time from one metastable state of a molecular system to another, which is inversely proportional to a chemical reaction rate \cite{Lelievre2010}.  Reaction rates are key inputs to kinetic Monte Carlo, a coarse-graining approach which reduces molecular dynamics to a jump process driven by transition probabilities \cite{Voter2005}. When the molecular dynamics are represented by high-fidelity quantum mechanical calculations, the cost associated with simulation-based estimates of reaction rates becomes intractable, particularly when the transition corresponds to a rare event \cite{Dellago1998, Allen2009}. Machine learning force fields \cite{Bartok2010, Drautz2019, Batzner2022, Musaelian2023}, which serve as surrogate models of the interatomic forces, are routinely employed to accelerate molecular dynamics simulation. In practice, these models are often validated by their static or equilibrium properties, which do not necessarily guarantee accurate dynamical behavior. Therefore, the essential task is to assess the accuracy of path-dependent, non-equilibrium properties of these models.

In this work, we introduce a novel objective function for learning a data-driven surrogate model of the drift function of an SDE, which enables the efficient and reliable estimation of a path-space observable such as an expected transition time. The learning objective is \textit{goal-oriented} in the sense that it targets the fidelity of a specific path-space observable of the stochastic dynamics. Moreover, it is a variational objective in that it learns the drift function of the SDE from a parametric family of functions and  relies only on data on the drift function of a reference high-fidelity SDE.  We show that the objective associated with goal-oriented learning offers advantages over prototypical choices of objectives, which have limited or no provable guarantees on the fidelity of the path-space observable corresponding to the learned SDE.

\subsection{Review of learning objectives for SDEs}
\label{sec:review}
To draw comparison with goal-oriented learning, we review common approaches for learning drift functions of first-order time-invariant SDEs, which can be classified by the nature of the available data and the corresponding learning objective. Consider the strong solution to an It\^o SDE on $m$-dimensional state space, 
    \[
        \textup{d} X_t = b(X_t) \textup{d}t + \textup{d}W_t \, , \ X_0=x,
    \]
where $b(x)$ denotes the drift function and $W_t$ denotes standard $m$-dimensional Brownian motion. Let $\tilde{b}(x; \theta) = \tilde{b}_\theta(x)$ be a surrogate of the drift function which belongs to a parametric model class indexed by $\theta \in \Theta$. The learning problem is to find a parameter $\theta^* \in \Theta$ which minimizes a prescribed loss function $\mathcal{L}$,
    \[
        \theta^* \in  \underset{\theta \in \Theta}{\text{argmin}} \ \mathcal{L}(\theta) \, .
    \]

\subsubsection{Regression of the drift function}
\label{subsec:regression}
In \textit{regression} approaches, the goal is to match the surrogate drift function to the reference drift function in a pointwise sense, given a dataset of reference drift evaluations at a finite collection of states, e.g., $\{X_i, b(X_i)\}_{i=1}^N$. A standard choice of loss function is the mean squared error in drifts, also referred to as the ``force matching'' objective \cite{Batzner2022, Musaelian2023},
    \begin{equation}
        \label{eq:fmloss}
        \mathcal{L}^{\textup{FM}}(\theta; b) \coloneqq \frac{1}{N} \sum^N_{i=1} \left\|b(X_i) - \tilde{b}_\theta (X_i) \right\|^2_2 \, .
    \end{equation}

For SDEs whose drift function is the negative gradient of a potential energy function $V$, e.g., $b(x) = -\nabla_x V(x)$, the surrogate drift function is typically modeled also as a gradient force, $\tilde{b}_\theta(x) = - \nabla_x \tilde{V}_\theta(x)$. Given data $\{X_i, V(X_i)\}_{i=1}^N$, one can then learn the surrogate drift function via regression of the potential energy function with the ``energy matching'' objective \cite{Batzner2022, Musaelian2023},
    \begin{equation}
        \label{eq:emloss}
        \mathcal{L}^{\textup{EM}}(\theta; V) \coloneqq \frac{1}{N} \sum^N_{i=1} \big| V(X_i) - \tilde{V}_\theta (X_i) \big|^2 \, .
    \end{equation}
If $V$ and $\tilde{V}_\theta$ are identical up to an additive constant and differentiable everywhere on a domain, then their first derivatives are identical on the domain and the force error $\frac{1}{N} \sum^N_{i=1} || \nabla_x V(X_i) - \nabla_x \tilde{V}_\theta (X_i) ||^2$ is also minimized.

\subsubsection{Minimization of discrepancy in probability measures}
\label{subsec:minmeasure}
In \textit{measure}-centric approaches, a model drift function is learned from a dataset of states distributed according to a probability measure associated with the SDE. These methods often apply to the setting where the drift function is unknown and must be inferred from observed states or paths.
When the states are ordered such that they correspond to equidistant observations of a strong solution of the reference SDE on $[0,T]$, i.e., $\{ X_{i} \}_{i=1}^N$ where $X_i=X_{i \Delta t}$ and $N\Delta t=T$, the likelihood of the states converges in the continuum limit to the the Radon-Nikodym derivative of the law of the model SDE, $\tilde{\mathbb{P}}^x_\theta$, with respect to the Wiener measure, $\mathbb{W}^x$ \cite[Sec. 5.3]{Pavliotis2014}. One can solve for the maximum likelihood estimate of the parameters by minimizing the negative log likelihood of the path, where the corresponding loss function is
    \begin{equation}
        \label{eq:mle}
        \mathcal{L}^{\textup{MLE}}(\theta; X) \coloneqq -\log \bigg( \frac{\textup{d}\tilde{\mathbb{P}}^x_\theta}{\textup{d}\mathbb{W}^x}(X) \bigg)  = - \frac{1}{2} \int_0^T |\tilde{b}_\theta(X_s)|^2 \textup{d}s - \int_0^T \tilde{b}_\theta(X_s)^\top \textup{d}W_s  \, .
    \end{equation}
Assuming the reference SDE is defined by a true drift parameter $\theta_0 \in \Theta$, the maximum likelihood estimate of the parameters is asymptotically unbiased and converges to $\theta_0$ as $T \to \infty$ \cite[Thm. 2.8]{Kutoyants2004}.  In this limit, the objective \eqref{eq:mle} relates to the Kullback--Leibler (KL) divergence between the laws of the reference and surrogate processes \cite{Dupuis2016}, such that the parameter estimate minimizes a discrepancy in path measures. One can consider a more general class of measure-centric approaches by replacing the KL divergence with another measure of discrepancy to implement as a loss function, e.g., 
\begin{equation}
        \label{eq:divloss}
        \mathcal{L}^{\mathbb{D}}(\theta; \rho) \coloneqq \mathbb{D}(\rho || \tilde{\rho}_\theta) \, ,
    \end{equation}
where $\mathbb{D}$ denotes a suitable measure of discrepancy between generic probability measures $\rho$ and $\tilde{\rho}_\theta$ of the reference and surrogate process. Broadly, measure-centric approaches target matching the surrogate SDE to the reference SDE in a distributional sense, such that their laws or probability densities are close with respect to the chosen discrepancy. A large body of work applies the form of \eqref{eq:divloss} to invariant measures of ergodic dynamical systems, given unordered states $\{X_i\}_{i=1}^N$ approximately i.i.d. with respect to the invariant measure, using discrepancies such as the Wasserstein-2 metric \cite{Muskulus2011,Yang2023} or maximum mean discrepancy \cite{Botvinick2024}. Given ordered states, \cite{Botvinick2024} demonstrates that the parameters may be uniquely identified by minimizing the discrepancy of invariant measures in time-delay coordinates. When applied to the Fisher divergence of time-marginal densities of the state, \eqref{eq:divloss} relates to the score matching objective used to train generative diffusion models \cite{Song2021mle}.

\subsubsection{Minimization of error in observables}
\label{subsec:minobs}
Surrogate models of the drift function may also be constructed using \textit{observable}-based approaches, utilizing datasets of noisy measurements of an observable $\mu$ which takes the form of an expectation of a functional of the stochastic dynamics. The aim is to maximize fidelity of the surrogate process in reproducing statistical properties of the reference process. The learning task can be formulated as an inverse problem, where the solution recovers the input parameters of a forward model which produced the observed outputs. In this setting, the forward model $Q$ is a map from parameters of the drift function of the surrogate SDE to its predicted observable, $Q(\theta) = \tilde{\mu}_{\theta}$, and noisy data of the observable $\{\mu^{(j)}\}_{j=1}^N$ are assumed to be perturbed by Gaussian noise $\{\eta^{(j)}\}_{j=1}^N$:
    \begin{equation}
        \label{eq:inverseprob}
        \mu^{(j)} = Q(\theta) + \eta^{(j)}, \quad \eta^{(j)} \sim N(0, \Gamma) \, .
    \end{equation}

The inverse problem of \eqref{eq:inverseprob} is closely related to least squares regression of the observable, where one minimizes the residual between the measured observable $\mu$ and the predicted observable $Q(\theta)$ scaled by the inverse of the noise covariance $\Gamma$, 
    \[
        \label{eq:inverseprobobj}
        \mathcal{L}^\text{inv}(\theta; \mu)= \frac{1}{2N} \sum^{N}_{j=1} \left\| \Gamma^{-1} (\mu^{(j)} - Q(\theta)) \right\|^2_2 \, .
    \]

This approach, also referred to as reverse Monte Carlo or inverse Monte Carlo, is adopted for radial distribution functions of molecular systems in \cite{Lyubartsev1995, Lyubartsev2002}. Observable-based approaches constitute a \textit{weak} approach to learning an SDE, in the sense that the model parameters are not uniquely determined; there can exist distinct SDEs which yield similar or identical observable values. This relaxation of uniqueness facilitates the existence of an admissible surrogate model whose purpose is to predict the unknown observable rather than to reproduce pathwise dynamics.

\subsection{Summary of contributions}
\label{sec:summary}

We propose a novel goal-oriented approach to surrogate modeling of SDEs which incorporates elements of each strategy reviewed in \Cref{sec:review}. As with the observable-based approaches of \Cref{subsec:minobs}, the ultimate goal is to minimize error in a specific observable, e.g.,
    \begin{equation}
        \label{eq:lossfunc}
        \theta^* \in  \underset{\theta \in \Theta}{\text{argmin}} \ \frac{1}{2} (\mu-\tilde{\mu}_\theta)^2 \, .
    \end{equation}
In many settings, the observable error on the right hand side of \eqref{eq:lossfunc} is intractable to compute, as high-quality ground truth data on the observable $\mu$ is inaccessible, whether due to high costs or infeasibility of experimentation or simulation. To overcome this data limitation, we instead rely only on drift data from the reference process, operating within the same data regime as the drift regression approaches of \Cref{subsec:regression}. Our idea is to leverage tools from the measure-centric approaches of \Cref{subsec:minmeasure} to derive a goal-oriented loss function $\mathcal{L}^\text{GO}$, dependent on the discrepancy between path measures, which upper bounds the error in the observable
    \[
        \label{eq:upperbd}
        \frac{1}{2} (\mu-\tilde{\mu}_\theta)^2 \leq \mathcal{L}^\text{GO}(\theta;b) \, .
    \]
Goal-oriented learning is then framed as a variational problem, where minimizing a tractable upper bound over a parametric class of drift functions is a proxy for minimizing the intractable error in the observable. The contributions of this paper can be summarized as follows: 
\begin{itemize}
    \item We prove an information-theoretic upper bound to the squared error in a path-space observable of the surrogate SDE in \Cref{lem:2momentbd}. The upper bound is useful for sensitivity analysis and uncertainty quantification, as well as for obtaining an efficient estimate of the order of magnitude of the error. In \Cref{sec:goerrbd}, we discuss how the upper bound holds under broad conditions on the observable.
    \item We introduce the goal-oriented loss in \Cref{def:goloss}, with both forward and reverse formulations to suit different data regimes. We identify conditions under which the goal-oriented loss holds as an upper bound to the error in the observable in \Cref{prop:golossineq}.
    \item We derive a closed-form gradient of the goal-oriented loss, to be used in gradient-based optimization, in \Cref{prop:gradgoloss}. A component of the formula is the gradient of the path-space observable with respect to parameters of the drift function, which we derive in \Cref{prop:gradobs}. 
    \item We present an empirical estimator of the goal-oriented loss in \eqref{eq:golossf_discrete} and highlight sources of numerical error in its evaluation. Through numerical experiments on estimating the mean first hitting time of overdamped Langevin systems, we show that models learned with the goal-oriented loss can exhibit faster convergence to the true value of the observable and greater robustness to error in the data distribution compared to other loss functions for surrogate modeling of SDEs. 
\end{itemize}

\subsection{Related work}

A number of prior works have developed bounds on the absolute error in path-space observables which hold under various conditions on the SDE and the observable. 

An error bound resulting from the Csisz\'ar-Kullback-Pinsker inequality is employed for parametric sensitivity analysis of molecular systems in \cite{Tsourtis2015} and for learning coarse-grained models in \cite{Harmandaris2016}. In \cite{Harmandaris2016}, it is stated that as a result of the bound, the minimization of relative entropy, i.e., KL divergence, is sufficient for controlling error over a broad class of path-space observables rather than a single specific observable. However, since the bound only holds under the condition that the path functional defining the path-space observable is either uniformly bounded or has a finite essential supremum, the learning objective based on relative entropy alone is not guaranteed to control error in observables defined by unbounded path functionals.

The work of \cite{Dupuis2016} develops an error bound which depends on the cumulant generating function of the path functional. A generalization of their bound is proposed in \cite{Branicki2021,Branicki2023}, expressed in terms of any information-monotone divergence between laws of the process. Both \cite{Dupuis2016} and \cite{Branicki2021} show that a linearized form of their bounds, derived from a Taylor expansion about the point where the error in the observable is zero, can be more tractable to compute; however, the linear approximation is not guaranteed to hold as an upper bound in instances where the error in the observable is large. Therefore, both works apply their bounds to sensitivity analysis and uncertainty quantification rather than learning tasks, which are deferred to future work in \cite{Branicki2023}. See \Cref{rmk:comparebd} for further comparison of our bound to the information inequalities in \cite{Harmandaris2016, Dupuis2016, Branicki2021}.

Whereas \cite{Harmandaris2016, Dupuis2016, Branicki2021} rely on information-theoretic methods, a different approach is taken in \cite{Zhang2021} of utilizing perturbation theory for ergodic Markov chains and applying their error bounds for identification of the drift and diffusion coefficients of the underlying SDE. Their analysis is limited to bounding error in state-space observables at the ergodic limit of the SDE, which does not include path-space observables describing non-ergodic properties of the stochastic process. 

Our work is the first to apply an error bound for expectations of unbounded path functionals to the variational learning of drift functions of SDEs, to explicitly derive the gradient of the error bound for the purpose of gradient-based optimization, and to provide numerical assessment of the robustness of goal-oriented learning schemes.

\medskip
The paper is organized as follows. \Cref{sec:sdes} introduces background on SDEs and relevant observables, as well as information divergences for path measures. \Cref{sec:golearning} discusses the main results as outlined above. \Cref{sec:numerics} presents numerical studies on identifying surrogate models of SDEs via goal-oriented learning with respect to a first hitting time statistic. \Cref{sec:conclusions} presents our conclusions.

\section{Stochastic differential equations}
\label{sec:sdes}

Throughout, let $||\cdot||$ denote the $\ell_2$ norm and $||\cdot||_1$ denote the $\ell_1$ norm on $\mathbb{R}^m$. For an arbitrary matrix $A \in \mathbb{R}^{m \times n}$, we denote the Frobenius norm by $||A||^2_F = \sum_{i,j} |A_{ij}|^2$, the Moore-Penrose pseudoinverse by $A^+$, and the matrix transpose by $A^\top$. Let $\mathcal{P}(U)$ denote the space of Borel probability measures on a measurable space $U$. For probability measures $p,q$ on $U$, $p \ll q$ denotes that $p$ is absolutely continuous with respect to $q$, and $p \sim q$ denotes that the two measures are mutually absolutely continuous.
\subsection{It\^o diffusion processes}

In this work, we consider a dynamical system governed by an It$\hat{\text{o}}$ SDE on some filtered probability space $(\Omega, \mathcal{F}, \{\mathcal{F}_t \}_{t \geq 0}, P)$,
    \begin{equation}
        \label{eq:sdefull}
        \textup{d}X_t = b(X_t) \, \textup{d}t + \sigma(X_t) \textup{d}W_t, \qquad X_0 = x \, .
    \end{equation}
The solution to the SDE is a Markov process $X(\omega) = ( X_t(\omega) )_{t \geq 0}$, where $X_t: \Omega \to \mathbb{R}^m$ is an $m$-dimensional continuous random variable, referred to as the ``state'' at time $t \geq 0$; $x \in \mathbb{R}^m$ is the initial condition; $b : \mathbb{R}^m \to \mathbb{R}^m$ is the time-invariant drift function; $\sigma : \mathbb{R}^m \to \mathbb{R}^{m \times n}$ is the time-invariant diffusion coefficient; and $W(\omega) = (W_t(\omega))_{t \geq 0}$, where $W_t: \Omega \to \mathbb{R}^n$, is standard $n$-dimensional Brownian motion. 

 The law of the solution to \eqref{eq:sdefull}, $\mathbb{P}^x \in \mathcal{P}(C([0,\infty), \mathbb{R}^m))$, is a probability measure over paths initialized at $x$. The Lebesgue probability density of the random variable $X_t$ at time $t$, $\rho_t \in L^1(\mathbb{R}^m)$, evolves according to the Fokker-Planck equation, 
    \begin{equation}
        \label{eq:fpk}
        \frac{\partial \rho_t}{\partial t} = \nabla \cdot \Big( -b \rho_t + \frac{1}{2} \nabla \cdot ( \Sigma \rho_t) \Big) \, ,
    \end{equation}
where $\Sigma(\cdot) = \sigma(\cdot) \sigma(\cdot)^\top$. We assume the following conditions: 
\begin{condition}[Unique strong solution and invariant measure]
    \label{cond:driftdiff}
    The SDE \eqref{eq:sdefull} admits an $(\mathcal{F}_t)_{t \geq 0}$-adapted strong solution $(X_t)_{t \geq 0}$ that is unique in law, and also a unique invariant measure $\pi \in \mathcal{P}(\mathbb{R}^m)$ whose density is the stationary solution to \eqref{eq:fpk}. 
\end{condition}

\begin{condition}[Semi-ellipticity, cf. \texorpdfstring{\cite[p.~289]{RevuzYor1999}}{}]
    \label{cond:elliptic}
    There exists a constant $\alpha > 0$ such that $0 \leq z^\top \Sigma(y)^{+} z \leq \alpha^{-2} ||z||^2$ for all $y, z \in \mathbb{R}^m$. 
\end{condition}

\subsection{Data-driven diffusion models}
\label{subsec:datadiffusion}

In many scientific applications of It\^o SDEs, the cost of evaluating the drift function $b$ renders it impractical to generate several samples from the path measure $\mathbb{P}^x$ or invariant measure $\pi$ of the SDE via numerical simulation. In such settings, one possible strategy is to learn a cheaper surrogate model of the drift function, $\tilde{b} : \mathbb{R}^m \to \mathbb{R}^m$, which defines a surrogate process $\tilde{X} = (\tilde{X}_t(\omega))_{t \geq 0}$ on $(\Omega, \mathcal{F}, \{\mathcal{F}_t \}_{t \geq 0}, P)$ according to 
    \begin{equation}
        \label{eq:sdesurrogate}
        \textup{d}\tilde{X}_t = \tilde{b}(\tilde{X}_t) \, \textup{d}t + \sigma(\tilde{X}_t) \textup{d}W'_t, \qquad \tilde{X}_0 = x \, ,
    \end{equation}
where $W'_t$ is an independent copy of the standard Brownian motion appearing in \eqref{eq:sdefull}.

To make comparisons between the reference process $X$ and surrogate process $\tilde{X}$, we utilize the Cameron-Martin-Girsanov change of measure. Let $u: \mathbb{R}^m \to \mathbb{R}^m$ be the difference in drift functions between the reference and surrogate process, 
\[
    \label{eq:scaleddiff}
    u(\cdot) \coloneqq b(\cdot) - \tilde{b}(\cdot) \, .
\]
Assume $u$ belongs to the vector space $U = \{v: \mathbb{R}^m \to \mathbb{R}^m | \ v \text{ is bounded} \}$. Define a continuous local martingale $(M^u_t)_{t \geq 0}$ with respect to any stochastic process $Y = (Y_t)_{t \geq 0} \in C([0,\infty), \mathbb{R}^m)$, 
\begin{equation}
    \label{eq:martingale}
    M^u_t(Y) \coloneqq \int_0^t (\sigma^+ u)(Y_s)^\top \textup{d}W_s \, ,
\end{equation}
and the quadratic variation process for $u \in U$,
    \begin{equation}
    \label{eq:quadvariation}
    \langle M^u, M^u \rangle_t (Y) = \int_0^t \Big( u^\top \Sigma^+ u\Big) (Y_s) \textup{d}s \, ,
    \end{equation}
where $\Sigma \coloneqq \sigma \sigma^\top$. Finally, define $(\mathcal{E}_t)_{t \geq 0}$ as
    \[
            \mathcal{E}_t(Y) \coloneqq \exp \big( - M^u_t(Y) - \frac{1}{2} \langle M^u, M^u \rangle_t (Y) \big) \, .
    \]

Under \Cref{cond:driftdiff} and the assumption that the difference in drift functions $u$ is bounded, \eqref{eq:sdesurrogate} has a unique weak solution \cite[Ch.~IX, Thm.~1.11]{RevuzYor1999}. We denote the law of the surrogate process $\tilde{X}$ as $\tilde{\mathbb{P}}^x \in \mathcal{P}(C([0,\infty), \mathbb{R}^m))$ and its invariant measure as $\tilde{\pi} \in \mathcal{P}(\mathbb{R}^m)$. 

Moreover, boundedness of $u$ and \Cref{cond:elliptic} imply that \eqref{eq:quadvariation} is almost surely finite and Novikov's criterion is met for any $Y$. This implies that for both $X$ and $\tilde{X}$ and any $t > 0$,
    \begin{equation}
    \label{eq:novikov}
    \mathbb{E}_{X \sim \mathbb{P}^x} \Big[ \exp  \Big( \frac{1}{2} \langle M^u, M^u \rangle_t(X) \Big) \Big] < \infty  \quad \text{and} \quad \mathbb{E}_{\tilde{X} \sim \tilde{\mathbb{P}}^x} \Big[ \exp  \Big( \frac{1}{2} \langle M^u, M^u \rangle_t(\tilde{X}) \Big) \Big] < \infty.
    \end{equation}
Novikov's criterion in \eqref{eq:novikov} yields that $(\mathcal{E}_s)_{0 \leq s \leq t}$ is a uniformly integrable martingale on the filtration $\{ \mathcal{F}_t \}_{t \geq 0}$ \cite[Ch.~VIII, Prop.~1.15]{RevuzYor1999}. By the Cameron-Martin-Girsanov theorem \cite[Thm.~8.6.8]{Oksendal2013}, $\mathcal{E}_t$ defines the Radon-Nikodym derivative of $\tilde{\mathbb{P}}^x$ w.r.t. $\mathbb{P}^x$. Evaluated at $X$, the strong solution to the SDE in \eqref{eq:sdefull}, the Radon-Nikodym derivative is 
        \begin{equation}
            \label{eq:rn_f}
            \frac{\textup{d}\tilde{\mathbb{P}}^x}{\textup{d}\mathbb{P}^x} \Big|_{\mathcal{F}_t}(X) = \mathcal{E}_t(X) = \exp \big( - M^u_t(X) - \frac{1}{2} \langle M^u, M^u \rangle_t (X) \big) \, .
        \end{equation}

From \eqref{eq:novikov} we also have that the reverse Radon-Nikodym derivative of $\mathbb{P}^x$ w.r.t. $\tilde{\mathbb{P}}^x$ is
        \begin{equation}
            \label{eq:rn_r}
            \frac{\textup{d}\mathbb{P}^x}{\textup{d}\tilde{\mathbb{P}}^x}\Big|_{\mathcal{F}_t} (X) = \mathcal{E}^{-1}_t(X) = \exp \big( M^u_t(X) + \frac{1}{2} \langle M^u, M^u \rangle_t (X) \big) \, .
        \end{equation}

The existence of the Radon-Nikodym derivative in both directions implies that $\tilde{\mathbb{P}}^x \sim \mathbb{P}^x$. 

The laws $\mathbb{P}^x$ and $\tilde{\mathbb{P}}^x$ are examples of probability measures over paths on the unbounded time domain $[0,\infty)$. One can also characterize the law of solutions on a bounded time domain. Let $\mathbb{P}_{[0,T]}^x, \tilde{\mathbb{P}}_{[0,T]}^x \in \mathcal{P}(C([0,T],\mathbb{R}^m))$ denote the finite-time path measures of $(X_t)_{t \in [0,T]}$ and $(\tilde{X}_t)_{t \in [0,T]}$, respectively, for some final time $0 < T < \infty$. Under the same conditions, we have $\mathbb{P}_{[0,T]}^x \sim \tilde{\mathbb{P}}_{[0,T]}^x$. We use the shorthand $\mathbb{E}^x[\cdot]$, $\tilde{\mathbb{E}}^x[\cdot]$, $\mathbb{E}^x_{[0,T]}[\cdot]$, and $\tilde{\mathbb{E}}^x_{[0,T]}[\cdot]$ to denote the expectation operator with respect to the measures $\mathbb{P}^x$, $\tilde{\mathbb{P}}^x$, $\mathbb{P}_{[0,T]}^x$, and $\tilde{\mathbb{P}}_{[0,T]}^x$, respectively.

\begin{remark}
 \label{remark_underdamped_Langevin_systems}
 The underdamped Langevin equation is an example of an SDE of the form in \eqref{eq:sdefull}, where $X_t=(q_t,p_t)\in\mathbb{R}^m\times\mathbb{R}^m$, 
 \begin{equation*}
  b(q_t,p_t)\coloneqq 
  \begin{pmatrix}
   M^{-1} p_t 
   \\
   -\nabla_q V(q_t)-\gamma M^{-1}p_t
  \end{pmatrix},
  \quad
  \sigma(q_t,p_t)\coloneqq
  \begin{pmatrix}
   0 & 0
   \\
  0 & \sqrt{2\gamma\beta^{-1}}\mathbb{I}_m
  \end{pmatrix},
 \end{equation*}
and $W_t\in \mathbb{R}^{2m}$ for every $t\geq 0$. The Moore-Penrose pseudoinverse of $\Sigma=\sigma\sigma^\top$ is constant and block matrix-valued, with zeros in all blocks except for the bottom right block, which is equal to $(2\gamma\beta^{-1})^{-1} \mathbb{I}_m$. Thus,  \Cref{cond:elliptic} holds with $\alpha^{-1}=(2\gamma\beta^{-1})^{-1/2}$. Next, suppose that $ q^\prime \mapsto \tilde{V}(q^\prime)\in\mathbb{R}$ is such that $q^\prime \mapsto -\nabla_q \tilde{V}(q^\prime)$ is bounded and measurable over the set of all admissible $q^\prime$, and that $\tilde{b}$ satisfies
 \begin{equation*}
  \tilde{b}(q_t,p_t) -b(q_t,p_t)=
  \begin{pmatrix} 
   0 
   \\
   -\nabla_q \tilde{V}(q_t).
  \end{pmatrix}
 \end{equation*}
Then the change of drift $\tilde{b}-b$ is bounded and measurable. Thus, if the reference underdamped Langevin system satisfies \Cref{cond:driftdiff}, then the laws of the reference and surrogate underdamped Langevin systems are mutually absolutely continuous.
\end{remark}

\subsection{SDE observables}

In applications, we are often interested in characterizing properties of the SDE in \eqref{eq:sdefull} using summary statistics, or expectations over probability measures of the SDE.  We refer generically to these properties as \textit{observables}. Much of prior literature concerns state-space observables -- expectations with respect to the invariant distribution of the state -- which describe properties of the system at equilibrium. Our work focuses on the more general case of path-space observables, which describe transient, non-equilibrium properties of the dynamics characterized by expectations with respect to the path measure of the SDE.

\begin{definition}[Path-space observable]
\label{def:pathobs}
A path-space observable $\mu \in \mathbb{R}$ is the expectation of a measurable functional $\phi_\tau : C([0,\infty), \mathbb{R}^m) \to \mathbb{R}$, 
    \[
        \mu \coloneqq \mathbb{E}^x[\phi_\tau] \, ,
    \]
where for any stochastic process $(Y_t)_{t \geq 0}$, $\phi_\tau(Y) = \int_0^\tau f(Y_s) \textup{d}s$ corresponds to the time integral of a particular function $f: \mathbb{R}^m \to \mathbb{R}$ up to a stopping time $\tau > 0$. 
\end{definition}

Following this definition, the path-space observable corresponding to the surrogate SDE in \eqref{eq:sdesurrogate} is $\tilde{\mu} \coloneqq \tilde{\mathbb{E}}^x[\phi_\tau] = \tilde{\mathbb{E}}^x \big[ \int_0^\tau f(\tilde{X}_s)\textup{d}s \big]$.

A canonical example of a path-space observable is the \textit{mean first exit time}, which characterizes transitions or reaction rates between metastable states of the stochastic process. Let $\tau$ denote the first time that the process escapes an open bounded set $B \subset \mathbb{R}^m$ with positive Lebesgue measure, 
    \begin{equation}
        \label{eq:firstexittime}
        \tau(Y) \coloneqq \inf \{ t \geq 0: Y_t \notin B\} \, .
    \end{equation}
Then, using the constant function $f=1$ in the definition of $\phi_\tau$ in \Cref{def:pathobs}, we have that $\phi_{\tau} = \tau$ and $\mu = \mathbb{E}^x[\tau]$ is the mean first exit time. When the diffusion coefficient of the SDE is $\sigma = \sqrt{2 \beta^{-1}}$ for some temperature parameter $\beta \in \mathbb{R}$, the scaled moment generating function of the first exit time $\tau$ can be written as
    \[
        \psi(x) = \mathbb{E}^x \big[ \exp( -\beta \epsilon \tau ) \big]
    \]
with some scale parameter $\epsilon$. From the Feynman-Kac formula \cite[Thm.~8.2.1]{Oksendal2013}, it can be shown that the scaled moment generating function is the solution to a Dirichlet boundary value problem, as defined in \cite{Hartmann2012} and \cite[Prop.~5.7.2]{Karatzas1991},
    \begin{equation}
        \label{eq:feynmankac}
        \beta^{-1} L \psi = \epsilon \psi, \quad \psi |_{\partial B} = 1 \, ,
    \end{equation}
where $L(\cdot) = \beta^{-1} \nabla_x^2 (\cdot) + b \cdot \nabla_x(\cdot)$ is the infinitesimal generator of the process. From the solution to the Feynman-Kac PDE in \eqref{eq:feynmankac} we can obtain estimates of the first and second moments of the first exit time from
\begin{equation}
    \label{eq:fkmoments}
    \mu = \mathbb{E}^x[\tau] = - \frac{1}{\beta} \frac{\textup{d}\psi}{\textup{d}\epsilon} \Big|_{\epsilon=0} \qquad \text{and} \qquad \mathbb{E}^x[\tau^2] = \frac{1}{\beta^2} \frac{\textup{d}^2\psi}{\textup{d}\epsilon^2} \Big|_{\epsilon=0} \, .
\end{equation}
    
In general, the Feynman-Kac PDE is tractable to solve only in low dimensions. Alternatively, one can compute a Monte Carlo estimate of a path-space observable from path realizations of the SDE. In particular, from a set of $N_\text{path}$ samples $\{X^{\delta,(i)}: i = 1,...,N_\text{path}\}$ independently generated via simulation of the SDE discretized with step size $\delta > 0$, $\mu$ can be estimated by
\[
    \mu \approx \hat{\mu} = \frac{1}{N_\text{path}} \sum_{i=1}^{N_\text{path}} \phi_\tau(X^{\delta,(i)}) \, .
\]

\subsection{Information theory with SDEs}
\label{sec:infotheory}

Information divergences between the path measures of SDEs are expressed in terms of the Radon-Nikodym derivative in \eqref{eq:rn_f} and \eqref{eq:rn_r}. We refer to the information divergence of $\mathbb{P}^x$ with respect to $\tilde{\mathbb{P}}^x$ as the \textit{forward} divergence and the information divergence of $\tilde{\mathbb{P}}^x$ with respect to $\mathbb{P}^x$ as the \textit{reverse} divergence.

\begin{definition}[Kullback-Leibler (KL) divergence]
\label{def:path_kl}
The KL divergence $\mathbb{D}_\textup{KL}$ of $\mathbb{P}^x$ w.r.t. $\tilde{\mathbb{P}}^x$ is defined as follows, where $\Sigma = \sigma \sigma^\top$:
    \[
        \label{eq:path_kl}
        \mathbb{D}_\textup{KL}(\mathbb{P}^x || \tilde{\mathbb{P}}^x) \coloneqq \int \log \Bigg( \frac{\textup{d}\mathbb{P}^x}{\textup{d}\tilde{\mathbb{P}}^x} \Bigg) \textup{d}\mathbb{P}^x = \mathbb{E}^x \Bigg[ \frac{1}{2} \int_0^\tau \left\| b - \tilde{b} \right\| ^2_{\Sigma^+} \textup{d}s \Bigg] \, .
    \]
\end{definition}

One can relate the asymptotic time average of KL divergence between finite-time path measures to the relative entropy rate $\mathcal{H}(\mathbb{P}^x || \tilde{\mathbb{P}}^x)$ as $T \to \infty$ \cite{Dupuis2016, Opper2017}.

\begin{definition}[Relative entropy rate]
\label{def:rer}
The relative entropy rate $\mathcal{H}$ of $\mathbb{P}^x$ w.r.t. $\tilde{\mathbb{P}}^x$ is
    \[
        \label{eq:rer}
        \mathcal{H}(\mathbb{P}^x || \tilde{\mathbb{P}}^x) \coloneqq \lim_{T \to \infty} \frac{1}{T} \mathbb{D}_\textup{KL}(\mathbb{P}_{[0,T]}^x || \tilde{\mathbb{P}}_{[0,T]}^x) = \frac{1}{2} \mathbb{E}_{\pi} \Big[ \left\| b - \tilde{b} \right\|_{\Sigma^+}^2 \Big] \, .
    \]
\end{definition}
Furthermore, we have the following relation from \cite[Lemma 3.2]{Dupuis2016}, when the initial conditions $X_0$ and $\tilde{X}_0$ are identical:
    \begin{equation}
        \label{eq:kl_rer}
        \mathbb{D}_\textup{KL}(\mathbb{P}_{[0,T]}^x || \tilde{\mathbb{P}}_{[0,T]}^x) = T \mathcal{H}(\mathbb{P}^x || \tilde{\mathbb{P}}^x) \, .
    \end{equation}

\newpage

\section{Main results}
\label{sec:golearning}

\subsection{Error bound on path-space observables}
\label{sec:goerrbd}

In \Cref{lem:2momentbd}, we introduce an information-theoretic upper bound to the absolute error in a path-space observable.

\begin{lemma}[Goal-oriented error bound]
\label{lem:2momentbd}
For $\phi_\tau \in L^2\big(\mathbb{P}^x\big) \cap L^2\big(\tilde{\mathbb{P}}^x\big)$, 
\begin{equation}
\label{eq:2momentbd}
|\mathbb{E}^x[\phi_\tau] - \tilde{\mathbb{E}}^x[\phi_\tau] | \leq \sqrt{ \mathbb{E}^x[\phi_\tau^2] + \tilde{\mathbb{E}}^x[\phi_\tau^2] } \sqrt{ 2 \mathbb{D}_\textup{KL} \big( \mathbb{P}^x || \tilde{\mathbb{P}}^x \big)} \, .
\end{equation}
\end{lemma}

\begin{proof}
    Let $f: Z \to E$ be an arbitrary measurable mapping, where $Z$ and $E$ are separable Banach spaces. In \cite[Lemma 6.37]{Stuart2010}, restated in \cite[Lemma 21]{Dashti2017}, it is shown that for any two probability measures $\nu$ and $\nu'$ that are absolutely continuous with respect to a measure $\upsilon$ on $Z$, the Hellinger distance $\mathbb{D}_\textup{H}(\nu||\nu')$ is well-defined and it holds that

    \begin{equation}
            \label{eq:stuartbd}
            || \mathbb{E}_\nu[f] - \mathbb{E}_{\nu'}[f] || \leq 2 \big( \mathbb{E}_\nu[||f||^2] + \mathbb{E}_{\nu'}[||f||^2]\big)^{\frac{1}{2}} \mathbb{D}_\textup{H}(\nu||\nu') \, ,
        \end{equation}
        provided that $f$ has second moments with respect to $\nu$ and $\nu'$. The proof given in \cite[p. 410]{Dashti2017} does not use the hypothesis that the domain is a Banach space, and is valid so long as $\nu$, $\nu'$, and $\upsilon$ are measures on a measurable space $(Z,\mathcal{Z})$. 
    
        \Cref{lem:2momentbd} is an application of \eqref{eq:stuartbd} where the measurable space is the path space $C([0,\infty),\mathbb{R}^{m})$ equipped with the cylindrical sigma-algebra, the probability measures are the laws of the SDEs, $\mathbb{P}^x$ and $\tilde{\mathbb{P}}^x$, and $f$ is taken to be the path functional $\phi_\tau: C([0,\infty),\mathbb{R}^m) \to \mathbb{R}$, such that
        \begin{equation}
            \label{eq:stuartbd_path}
            |\mathbb{E}^x[\phi_\tau] - \tilde{\mathbb{E}}^x[\phi_\tau] | \leq 2 \big( \mathbb{E}^x[\phi_\tau^2] + \tilde{\mathbb{E}}^x[\phi_\tau^2] \big)^{\frac{1}{2}}\mathbb{D}_\textup{H} \big( \mathbb{P}^x || \tilde{\mathbb{P}}^x \big) \, .
        \end{equation}
    
    The final expression is recovered by bounding the Hellinger distance of the path measures from above by the KL divergence (\Cref{def:path_kl}), using the information inequality $\mathbb{D}^2_\textup{H}(\mathbb{P}^x || \tilde{\mathbb{P}}^x ) \leq \frac{1}{2} \mathbb{D}_\textup{KL}(\mathbb{P}^x || \tilde{\mathbb{P}}^x )$ \cite[Lemma 22]{Dashti2017}.
\end{proof}

To our knowledge, our work makes the first application of \cite[Lemma 6.37]{Stuart2010} to expectations of path functionals of solutions of SDEs with respect to their laws. The formula for the Hellinger distance between path measures requires the computation of a stochastic integral, which is a random variable dependent on the Wiener process and therefore more difficult to numerically estimate and analyze compared to deterministic integrals; see \cite{Kloeden1992}. Therefore, we opt to utilize KL divergence as the information divergence between $\mathbb{P}^x$ and $\tilde{\mathbb{P}}^x$, since the formula for KL divergence in \Cref{def:path_kl} does not involve computing a stochastic integral.

The bound inherits properties of an information divergence; namely, that it is non-negative for all $\mathbb{P}^x$ and $\tilde{\mathbb{P}}^x$, and for $\phi_\tau$ which is almost surely nonzero with respect to $\mathbb{P}^x$ and $\tilde{\mathbb{P}}^x$, the upper bound is zero if and only if $\mathbb{P}^x = \tilde{\mathbb{P}}^x$. It holds for a broad class of observables, requiring only square integrability of the path functional under $\mathbb{P}^x$ and $\tilde{\mathbb{P}}^x$. The bound is nontrivial as long as $\mathbb{P}^x \ll \tilde{\mathbb{P}}^x$. Moreover, it does not assume closeness of the path measures in terms of a divergence metric, as do the linearized bounds in \cite{Dupuis2016, Branicki2021}.

\begin{remark}
    \label{rmk:comparebd}
    Various information inequalities which upper bound the error in path-space observables have been proposed in literature. Let $\mathbb{Q}^x \in \mathcal{P}(C([0,\infty), \mathbb{R}^m))$ be a dominating measure satisfying $\mathbb{P}^x \ll \mathbb{Q}^x$ and $\tilde{\mathbb{P}}^x \ll \mathbb{Q}^x$. For path functionals with finite essential supremum norm, e.g., $\phi_\tau \in L^\infty(\mathbb{Q}^x)$, the bound based on the Csisz\'ar-Kullback-Pinsker inequality in \cite{Harmandaris2016} is
    \begin{equation}
    \label{eq:esssupbd}
    |\mathbb{E}^x[\phi_\tau] - \tilde{\mathbb{E}}^x[\phi_\tau] | \leq || \phi_\tau||_{L^\infty(\mathbb{Q}^x)} \sqrt{2 \mathbb{D}_\textup{KL}\big( \mathbb{P}^x || \tilde{\mathbb{P}}^x \big)} \, .
    \end{equation}
    
    The error bound in \cite[Thm.~3.3]{Dupuis2016} relies on the cumulant generating function of the centered observable functional, requiring that $\Lambda(\lambda; \tilde{\mathbb{P}}^x) \coloneqq \log \tilde{\mathbb{E}}^x\Big[ e^{\lambda \big( \phi_\tau-\tilde{\mathbb{E}}^x[\phi_\tau] \big)} \Big]$ is finite for all $\lambda$ in some open neighborhood of the origin. Then, one has the following bound: 
    \begin{equation}
        \label{eq:cgfbd}
        |\mathbb{E}^x[\phi_\tau] - \tilde{\mathbb{E}}^x[\phi_\tau] | \leq \inf_{\lambda > 0} \frac{1}{\lambda} \bigg\{ \Lambda(\lambda; \tilde{\mathbb{P}}^x) + \mathbb{D}_\textup{KL}\big( \mathbb{P}^x || \tilde{\mathbb{P}}^x \big) \bigg\} \, .
    \end{equation}

The bound in \eqref{eq:cgfbd} is closely related to one proposed in \cite[Thm.~3.1]{Branicki2021}, expressed in terms of any $\varphi$-divergence metric, which holds so long as $\phi_\tau$ belongs to the Orlicz subspace associated with the $\varphi$-functional. The conditions for both of the bounds in \eqref{eq:cgfbd} and \cite{Branicki2021} imply that $\phi_\tau$ must be exponentially integrable with respect to $\tilde{\mathbb{P}}^x$, e.g., $\phi_\tau \in L^\text{exp}(\tilde{\mathbb{P}}^x) \coloneqq \{ \phi : C([0,\infty), \mathbb{R}^m) \to \mathbb{R} \ | \ \exists \ \lambda_0 > 0 \ s.t. \ \Lambda(\pm \lambda_0; \tilde{\mathbb{P}}^x)< \infty\}$. 

In goal-oriented learning and other applications, it is often favorable for an error bound on the path-space observable to be symmetric with respect to the order of the path measures. One can verify that the conditions required on $\phi_\tau$ for each of the above bounds to have this symmetry are stricter in comparison to those of \Cref{lem:2momentbd}. In particular, one requires that $\phi_\tau \in L^\infty\big(\mathbb{P}^x\big) \cap L^\infty(\tilde{\mathbb{P}}^x\big)$\footnote{It is straightforward to show all $\phi_\tau \in L^\infty(\mathbb{Q}^x)$ also belong to $L^\infty(\mathbb{P}^x) \cap L^\infty(\tilde{\mathbb{P}}^x)$, since any set with $\mathbb{Q}^x$-measure zero also has zero measure under  $\mathbb{P}^x$ and $\tilde{\mathbb{P}}^x$.} for \eqref{eq:esssupbd} to hold and $L^\text{exp}\big(\mathbb{P}^x\big) \cap L^\text{exp}\big(\tilde{\mathbb{P}}^x\big)$ for symmetry of the bound in \eqref{eq:cgfbd}. Since we have $L^p(\mathbb{U}) \subset L^2(\mathbb{U})$ for any probability measure $\mathbb{U}$ and $p \in (2,\infty)$, and $L^\text{exp}(\mathbb{U}) \subseteq \bigcap_{p=1}^\infty L^p(\mathbb{U})$, the space $L^2\big(\mathbb{P}^x\big) \cap L^2\big(\tilde{\mathbb{P}}^x\big)$ is strictly larger than $L^\infty\big(\mathbb{P}^x\big) \cap L^\infty(\tilde{\mathbb{P}}^x\big)$ as well as $L^\text{exp}\big(\mathbb{P}^x\big) \cap L^\text{exp}(\tilde{\mathbb{P}}^x\big)$. Therefore, the bound in \Cref{lem:2momentbd} can be stated with $\mathbb{P}^x$ and $\tilde{\mathbb{P}}^x$ in either order under strictly weaker conditions on $\phi_\tau$ than those required for \eqref{eq:esssupbd} or \eqref{eq:cgfbd} to hold for both orderings of the path measures. 

\Cref{lem:2momentbd} is derived from the Hellinger distance-based inequality in \cite[Lemma 6.37]{Stuart2010}, restated in \eqref{eq:stuartbd}. Improvements to this inequality are proposed in \cite[Lemma A.1]{Katsoulakis2017} and \cite[Lemma 3]{Cui2023}, which achieve a tighter bound by using an optimal centering of the functional $\phi_\tau$. However, these improved bounds do not necessarily yield a more effective loss function for goal-oriented learning. Since they depend on the variance of $\phi_\tau$ instead of the second moment of $\phi_\tau$ under $\tilde{\mathbb{P}}^x$, the gradient of the improved bounds with respect to perturbations to the surrogate measure $\tilde{\mathbb{P}}^x$ contain additional terms compared to the gradient of \cite[Lemma 6.37]{Stuart2010}. Each additional term introduces independent Monte Carlo error, which increases the variance of the resulting gradient estimator. This problem is consistent with the observation in \cite{Rainforth2018} that tighter variational bounds may be associated with a lower signal-to-noise ratio in their gradient estimators, which reduces the efficacy of stochastic gradient descent. For these reasons, we adopt the bound in \cite[Lemma 6.37]{Stuart2010} instead of its variants in \cite{Katsoulakis2017,Cui2023} for developing a goal-oriented loss. 
\end{remark}

\begin{remark}
\label{rmk:hittime}
The mean first exit time of a stochastic process is an example of a path-space observable which can satisfy the integrability condition of \Cref{lem:2momentbd} but not that of the bound in \eqref{eq:esssupbd}, since the first exit time on an unbounded time domain defined in \eqref{eq:firstexittime} does not have a finite essential supremum in general. Nevertheless, the first exit time has a finite second moment with respect to $\mathbb{P}^x$ and $\tilde{\mathbb{P}}^x$ under well-known regularity assumptions for the It\^o SDEs in \eqref{eq:sdefull} and \eqref{eq:sdesurrogate}. For instance, if the mean first exit time is almost surely finite, the SDE is uniformly elliptic in the bounded domain $B$ in \eqref{eq:firstexittime}, and it has a unique weak solution, then \eqref{eq:feynmankac} has a classical solution which yields the moments of the first exit time from $B$ \cite[Prop. 5.7.2]{Karatzas1991}. More general sufficient conditions for square integrability of the first exit time of SDEs, including those with degenerate or time-dependent coefficients, are provided in \cite[Lemma 2.1]{Feng2021}.
\end{remark}

\subsection{Goal-oriented loss functions}
\label{sec:goloss}

Taking a variational learning approach, we adapt the error bound in \Cref{lem:2momentbd} into a loss function for goal-oriented learning of the drift function of an SDE, where the ``goal'' is to control error in the path-space observable predicted by the learned model. Let $\Theta$ be a nonempty open subset of $\mathbb{R}^d$. For $\theta \in \Theta$, the surrogate SDE associated with a surrogate drift function $\tilde{b}_\theta$ is
\begin{equation}
    \label{eq:sdeparametric}
    \textup{d}\tilde{X}^\theta_t = \tilde{b}_\theta(\tilde{X}^\theta_t) \textup{d}t + \sigma(\tilde{X}^\theta_t) \textup{d}W_t, \quad \tilde{X}^\theta_0 = x \, .
\end{equation}
Let $\tilde{b}_\theta$ belong to the model class $\mathcal{B}$, indexed by parameters $\theta \in \Theta$, defined as
\begin{align*}
    \mathcal{B} \coloneqq \{ \tilde{b}_\theta : \mathbb{R}^m \to \mathbb{R}^m \ | & \ \theta \in \Theta, b-\tilde{b}_\theta \text{ is bounded and measurable,} \\
    & \text{\eqref{eq:sdeparametric} admits a unique invariant measure }\tilde{\pi}_\theta \} \, .
\end{align*}
From the condition that $b-\tilde{b}_\theta$ is bounded, \eqref{eq:sdeparametric} is guaranteed a unique weak solution with law $\tilde{\mathbb{P}}^x_\theta \in \mathcal{P}(C([0,\infty), \mathbb{R}^m))$, where  $\tilde{\mathbb{P}}^x_\theta \sim \mathbb{P}^x$ follows from Novikov's criterion; see the discussion in \Cref{subsec:datadiffusion}. It is assumed that the map $\theta \mapsto \tilde{b}_\theta$ is injective, i.e., elements of $\mathcal{B}$ are uniquely determined by their associated parameter $\theta$. 

\begin{definition}[Goal-oriented loss functions]
\label{def:goloss}
For $\phi_\tau : C([0,\infty),\mathbb{R}^m) \to \mathbb{R}$ defined by a random stopping time $\tau > 0$ and $\mathcal{M} \geq 0$, define the forward goal-oriented loss $\mathcal{L}_\tau^{\textup{GO}}: \Theta \to \mathbb{R}$ and the reverse goal-oriented loss $\reverse{\mathcal{L}}_\tau^{\textup{GO}}: \Theta \to \mathbb{R}$ as
\begin{subequations}
\label{eq:goloss}
\begin{alignat}{4}
    \label{eq:golossf_tau}
    \mathcal{L}_\tau^{\textup{GO}}(\theta) \coloneqq \Big(\mathcal{M} + \tilde{\mathbb{E}}^x_\theta[\phi_\tau^2] \Big) \mathbb{D}_{\textup{KL}}(\mathbb{P}^x||\tilde{\mathbb{P}}^x_\theta) \, ,
    \\
    \label{eq:golossr_tau}
    \reverse{\mathcal{L}}_\tau^{\textup{GO}}(\theta) \coloneqq \Big(\mathcal{M} + \tilde{\mathbb{E}}^x_\theta[\phi_\tau^2] \Big) \mathbb{D}_{\textup{KL}}(\tilde{\mathbb{P}}^x_\theta || \mathbb{P}^x) \, .
\end{alignat}
\end{subequations}
For $\phi_T : C([0,T],\mathbb{R}^m) \to \mathbb{R}$ defined by a deterministic time $0 < T < \infty$, the forward and reverse goal-oriented losses reduce to $\mathcal{L}_T^{\textup{GO}}: \Theta \to \mathbb{R}$ and $\reverse{\mathcal{L}}_T^{\textup{GO}}: \Theta \to \mathbb{R}$, where
\begin{subequations}
\label{eq:goloss_T}
\begin{alignat}{4}
    \label{eq:golossf}
    \mathcal{L}_T^{\textup{GO}}(\theta) \coloneqq T \Big(\mathcal{M} + \tilde{\mathbb{E}}^x_{\theta,[0,T]}[\phi_T^2] \Big) \mathcal{H}(\mathbb{P}^x||\tilde{\mathbb{P}}^x_\theta) \, ,
    \\
    \label{eq:golossr}
    \reverse{\mathcal{L}}_T^{\textup{GO}}(\theta) \coloneqq T \Big(\mathcal{M} + \tilde{\mathbb{E}}^x_{\theta,[0,T]}[\phi_T^2] \Big) \mathcal{H}(\tilde{\mathbb{P}}^x_\theta || \mathbb{P}^x) \, .
\end{alignat}
\end{subequations}
\end{definition}

Under appropriate conditions, both the forward and reverse goal-oriented losses are upper bounds to the squared error in the path-space observable.

\begin{proposition}[Goal-oriented losses bound error in the observable]
\label{prop:golossineq}
For $\phi_\tau \in L^2\big(\mathbb{P}^x\big) \cap L^2\big(\tilde{\mathbb{P}}^x_\theta\big)$ and $\mathcal{M} \geq \mathbb{E}^x[\phi^2_\tau]$,
\[
    \frac{1}{2}(\mathbb{E}^x[\phi_\tau]-\tilde{\mathbb{E}}^x_\theta[\phi_\tau])^2 \leq \mathcal{L}_\tau^{\textup{GO}}(\theta) \quad \text{and} \quad \frac{1}{2}(\mathbb{E}^x[\phi_\tau]-\tilde{\mathbb{E}}^x_\theta[\phi_\tau])^2 \leq \reverse{\mathcal{L}}_\tau^{\textup{GO}}(\theta).
\]
For $\phi_T \in L^2\big(\mathbb{P}^x_{[0,T]}\big) \cap L^2\big(\tilde{\mathbb{P}}^x_{\theta,[0,T]}\big)$ and $\mathcal{M} \geq \mathbb{E}^x_{[0,T]}[\phi^2_T]$,
\[
    \frac{1}{2}(\mathbb{E}^x_{[0,T]}[\phi_T]-\tilde{\mathbb{E}}^x_{\theta,[0,T]}[\phi_T])^2 \leq \mathcal{L}_T^{\textup{GO}}(\theta) \quad \text{and} \quad \frac{1}{2}(\mathbb{E}^x_{[0,T]}[\phi_T]-\tilde{\mathbb{E}}^x_{\theta,[0,T]}[\phi_T])^2 \leq \reverse{\mathcal{L}}_T^{\textup{GO}}(\theta).
\]
    
\end{proposition}
\begin{proof} We first consider the case of path functionals with a random stopping time $\tau \geq 0$. The inequality in \Cref{lem:2momentbd} is maintained if we replace $\mathbb{E}^x[\phi_\tau^2]$ with an upper bound $\mathcal{M}$ which is constant with respect to $\theta$. Since both sides of the inequality are strictly non-negative, squaring preserves order and we have a bound on the squared error in the path-space observable of the form of $\mathcal{L}^\textup{GO}_\tau$.  The upper bound with $\reverse{\mathcal{L}}^\textup{GO}_\tau$ follows from the fact that \Cref{lem:2momentbd} holds with a reverse ordering of path measures.

The result for path functionals with a deterministic stopping time are derived in a similar fashion, using the fact that \Cref{lem:2momentbd} holds for mutually absolutely continuous path measures on $[0,T]$. We substitute $\mathbb{D}_\textup{KL} \big( \mathbb{P}^x_{[0,T]} || \tilde{\mathbb{P}}^x_{\theta,[0,T]}  
\big) \allowbreak = T\mathcal{H}(\mathbb{P}^x || \tilde{\mathbb{P}}^x_{\theta})$ from \eqref{eq:kl_rer} to obtain the forms of $\mathcal{L}_T^\textup{GO}$ and $\reverse{\mathcal{L}}_T^\textup{GO}$.
\end{proof}

\medskip 

Given the error guarantees of \Cref{prop:golossineq}, minimizing the goal-oriented loss in \Cref{def:goloss} serves as a proxy for minimizing error in the observable. By introducing the constant $\mathcal{M}$, we show that it suffices to know an upper bound to the second moment of the observable rather than its exact value, which does not depend on the parameter $\theta$. In practice, $\mathcal{M}$ can be treated as a scale factor in the loss which adjusts the relative weight of the observable-dependent summand and observable-independent summand -- e.g., $\tilde{\mathbb{E}}^x_\theta[\phi_\tau^2] \mathbb{D}_{\textup{KL}}(\mathbb{P}^x||\tilde{\mathbb{P}}^x_\theta)$ and $\mathcal{M} \mathbb{D}_{\textup{KL}}(\mathbb{P}^x||\tilde{\mathbb{P}}^x_\theta)$ of \eqref{eq:goloss}, respectively. 

The goal-oriented losses \eqref{eq:golossf_tau} and \eqref{eq:golossr_tau} apply to the more general case of path-space observables defined on unbounded time domains. In the case of path-space observables defined on bounded time domains, \eqref{eq:golossf} and \eqref{eq:golossr} often have reduced computation and data requirements. This is because Monte Carlo estimators of the relative entropy rate require only independent samples of the invariant density, whereas Monte Carlo estimators of the KL divergence for finite-time path measures require multiple path realizations of the SDE. Since only finite-time simulations are possible in practice, the remainder of the paper focuses on the forms of the goal-oriented loss in \eqref{eq:golossf} and \eqref{eq:golossr}.

The forward and reverse formulations of the goal-oriented loss differ only in the order of arguments of the discrepancy measure, which determines the underlying probability measure for the associated expectation. Therefore, the two formulations suit different data scenarios for training: the forward loss $\mathcal{L}_T^{\textup{GO}}$ is appropriate when given sample states from $\pi$ (or $\mathcal{L}^{\textup{GO}}_\tau$ given sample paths from $\mathbb{P}^x$), whether generated from prior simulations or acquired from open-source repositories. In this setting, the data remain fixed across training iterations; however, such samples from the reference process are often difficult to obtain and may poorly represent the reference measure. The reverse loss $\reverse{\mathcal{L}}_T^{\textup{GO}}$ relies on sample states from $\tilde{\pi}_\theta$ (or $\reverse{\mathcal{L}}^{\textup{GO}}_\tau$ on sample paths from $\tilde{\mathbb{P}}^x_\theta$), which are more feasible to generate from simulation due to the low cost of the surrogate model. The reverse loss incurs a sampling step for each iteration of the surrogate model during training, but offers the advantage that sampling error can be systematically reduced through longer repeated simulations. The algorithms for forward and reverse goal-oriented learning are in \Cref{supp:alg}. 

With either formulation, we solve for the optimal parameters which minimize the goal-oriented loss, e.g., $\theta^* =\text{argmin}_\theta \ \mathcal{L}_T^\textup{GO}(\theta)$ or $\reverse{\theta}^* = \text{argmin}_\theta \ \reverse{\mathcal{L}}_T^\textup{GO}(\theta)$. In general, we have that $\theta^* \neq \reverse{\theta}^*$ due to asymmetry of the relative entropy rate. Moreover, $\text{argmin}_\theta \ \mathcal{L}_T^\textup{GO}(\theta)$ and $\text{argmin}_\theta \ \reverse{\mathcal{L}}_T^\textup{GO}(\theta)$ are not guaranteed to coincide with $\text{argmin}_\theta \ \frac{1}{2} (\mu-\tilde{\mu}_\theta)^2$. This fact is a typical result of variational learning, where a controlled mismatch between a tractable upper bound and the true objective is accepted in order to yield useful approximations of the true minimizer. Finally, the goal-oriented loss admits multiple minimizers, since the error in the observable does not uniquely determine the SDE. Among multiple admissible models, the goal-oriented loss can be regularized so as to favor those with additional properties, such as reduced variance in the observable or increased fidelity to strong solutions of the reference SDE.

\subsection{Gradient of the goal-oriented loss function}
\label{sec:gradloss}

The gradient of the goal-oriented loss function can be employed in gradient-based optimization algorithms to locate critical points $\hat{\theta} \in \Theta$, i.e., parameters which satisfy $\nabla_\theta \mathcal{L}_T^\textup{GO}(\hat{\theta})=0$ or $\nabla_\theta \reverse{\mathcal{L}}_T^\textup{GO}(\hat{\theta})=0$. 
We show that the gradient of the goal-oriented losses in \eqref{eq:golossf} and \eqref{eq:golossr} with respect to parameters have closed-form expressions which are tractable for computation. First, we state additional assumptions which are required to derive an analytical expression for the gradient.

\begin{condition}[Differentiability]
    \label{cond:differentiable}
    The map $\theta \mapsto \tilde{b}_\theta \in \mathcal{B}$ is continuously differentiable for $\theta \in \Theta$, and the partial derivatives $\partial_{\theta_i} \tilde{b}_\theta: \mathbb{R}^m \to \mathbb{R}^m$ for $i=1,...,d$ are bounded. 
\end{condition}

\begin{condition}[Exponentially integrable stopping time]
    \label{cond:observable}
    The path-space observable $\tilde{\mu}(\theta) = \tilde{\mathbb{E}}^x_\theta[\phi_\tau]$ is defined by some functional $\phi_\tau$ where for every $\theta \in \Theta$, there exists a $\lambda_0 > 0$ such that $\tilde{\mathbb{E}}^x_\theta[\exp(\lambda_0 \tau)] < \infty$. 
\end{condition}

\begin{condition}[Existence of dominating functions]
    \label{cond:domfcn}
    There exist functions $\psi_1, \psi_2, \allowbreak \psi_3 \in L^1$ where for almost every $x$, $\sup_{\theta \in \Theta} || \tilde{b}_\theta(x)||_1 \leq \psi_1(x)$, $\sup_{\theta \in \Theta} || \nabla_\theta \tilde{b}_\theta(x)||_F \leq \psi_2(x)$, and $\sup_{\theta \in \Theta} \allowbreak || \nabla_\theta \tilde{\pi}_\theta(x)||_1 \allowbreak \leq \psi_3(x)$. 
\end{condition}

Next, we present the gradient of the path-space observable. This result uses previous work on the Fr\'echet derivative of expected path functionals \cite{Lie2021}. We state the proof in \Cref{supp:gradobspf}.

\begin{proposition}[Gradient of path-space observables]
\label{prop:gradobs}
 Assume Conditions \ref{cond:elliptic}, \ref{cond:differentiable}, and \ref{cond:observable} hold. Then the partial derivative of $\tilde{\mu}$ w.r.t. $\theta_i$ for $i=1,...,d$ is
\begin{equation}
    \label{eq:partial}
    \partial_{\theta_i} \tilde{\mu}(\theta) = \tilde{\mathbb{E}}^x_\theta[\phi_\tau M^{\partial_{\theta_i} \tilde{b}_\theta}_\tau] = \tilde{\mathbb{E}}^x_\theta \Big[ \phi_\tau \int_0^\tau (\sigma^+\partial_{\theta_i} \tilde{b}_\theta)^\top \textup{d}W_s\Big]\, . 
\end{equation}
\end{proposition}

With \Cref{prop:gradobs}, we can write the gradient of the goal-oriented losses \eqref{eq:golossf} and \eqref{eq:golossr} in \Cref{prop:gradgoloss}, whose proof is given in \Cref{supp:proppf}.

\begin{proposition}[Gradient of goal-oriented loss functions]
\label{prop:gradgoloss}
Assume Conditions \ref{cond:elliptic}, \ref{cond:differentiable}, and \ref{cond:domfcn} hold. Let $J_\theta: \mathbb{R}^m \to \mathbb{R}^{m \times d}$ be the Jacobian matrix $J_\theta = \nabla_\theta\tilde{b}_\theta$; i.e., for $y \in \mathbb{R}^m$ and $\theta' \in \Theta$,
\begin{equation}
\label{eq:jacobian}
J_{\theta'}(y) = \begin{bmatrix}
    \frac{\partial (\tilde{b}_\theta)_1}{\partial \theta_1}(y) \cdots \frac{\partial (\tilde{b}_\theta)_1}{\partial \theta_d}(y) \\
    \vdots \quad \ddots \quad \vdots \\
    \frac{\partial (\tilde{b}_\theta)_m}{\partial \theta_1}(y) \cdots \frac{\partial (\tilde{b}_\theta)_m}{\partial \theta_d}(y)
    \end{bmatrix}_{\theta=\theta'}\, .
\end{equation}

The gradient of the forward goal-oriented loss \eqref{eq:golossf} with respect to $\theta$ is
\begin{equation}
\label{eq:golossf_grad}
        \nabla_\theta \mathcal{L}_T^\textup{GO} (\theta) = \frac{T}{2} \Big( \mathcal{H}(\mathbb{P}^x || \tilde{\mathbb{P}}^x_\theta) \nabla_\theta \tilde{\mathbb{E}}^x_{\theta,[0,T]}[\phi_T^2] + \big(\mathcal{M} + \tilde{\mathbb{E}}^x_{\theta,[0,T]}[\phi_T^2] \big) \nabla_\theta \mathcal{H}(\mathbb{P}^x || \tilde{\mathbb{P}}^x_\theta) \Big) \, ,
\end{equation}       
where the gradient of the forward relative entropy rate is
\begin{equation}
    \label{eq:rerf_grad}
    \nabla_\theta \mathcal{H}(\mathbb{P}^x || \tilde{\mathbb{P}}^x_\theta) = \mathbb{E}_\pi \Big[ J_\theta^\top \Sigma^+ (\tilde{b}_\theta - b)  \Big] \, .
\end{equation}
The gradient of the reverse goal-oriented loss \eqref{eq:golossr} with respect to $\theta$ is 
\begin{equation}
        \label{eq:golossr_grad}
        \nabla_\theta \reverse{\mathcal{L}}_T^\textup{GO} (\theta) = \frac{T}{2} \Big(  \mathcal{H}(\tilde{\mathbb{P}}^x_\theta || \mathbb{P}^x) \nabla_\theta \tilde{\mathbb{E}}^x_{\theta,[0,T]}[\phi_T^2] + \big(\mathcal{M} + \tilde{\mathbb{E}}^x_{\theta,[0,T]}[\phi_T^2] \big) \nabla_\theta \mathcal{H}(\tilde{\mathbb{P}}^x_\theta || \mathbb{P}^x) \Big) ,
\end{equation}
where the gradient of the reverse relative entropy rate is
\begin{equation}
\begin{aligned}
    \label{eq:rerr_grad}
    \nabla_\theta \mathcal{H}( \tilde{\mathbb{P}}^x_\theta || \mathbb{P}^x) = 
    \frac{1}{2} \mathbb{E}_{\tilde{\pi}_\theta} \Big[ || \sigma^+(b - \tilde{b}_\theta) ||_2^2 \ \nabla_\theta \log \tilde{\pi}_\theta + 2 J_\theta^\top \Sigma^+(\tilde{b}_\theta - b) \Big] \, . 
\end{aligned}
\end{equation}

The gradient of the second moment of the path-space observable in both \eqref{eq:golossf_grad} and \eqref{eq:golossr_grad} is
    \begin{equation}
        \label{eq:2momentgrad}
        \nabla_\theta \tilde{\mathbb{E}}^x_{\theta,[0,T]} [\phi_T^2] = \tilde{\mathbb{E}}^x_{\theta,[0,T]}[\phi_T^2 M_T^{J_\theta}] \, ,
    \end{equation}
where $(M_t^{J_\theta})_{t \geq 0}$ with respect to a stochastic process $(Y_t)_{t \geq 0}$ is the $d$-dimensional local martingale
    \begin{equation}
        \label{eq:dmartingale}
        M^{J_\theta}_t(Y) \coloneqq \int_0^t (\sigma^+J_\theta)(Y_s)^\top \textup{d}W_s \, .
    \end{equation}

\end{proposition}

A gradient estimator using the closed-form expression offers several advantages. By \Cref{prop:gradgoloss}, the gradient of the goal-oriented loss can be written exactly as a sum of products of expectations, allowing one to use the same set of sample paths and states to calculate Monte Carlo estimates of both the loss and its gradient. The variance of the closed-form gradient estimator is based only on finite-sample Monte Carlo error, unlike finite-difference gradient estimators whose variance is often amplified by small step sizes. Moreover, one can apply variance reduction techniques for estimating expectations, such as importance sampling or control variates, to reduce the variance of the closed-form gradient estimator. The gradient calculation is linear in cost to the number of parameters and easily parallelizable across parameter dimensions.

\subsection{Learning with inexact losses}
\label{sec:inexactloss}

In practice, the goal-oriented losses, e.g., \eqref{eq:golossf} and \eqref{eq:golossr}, and their gradients, \eqref{eq:golossf_grad} and \eqref{eq:golossr_grad}, are computed using empirical estimators based on the numerical simulation of SDEs. One can generate a numerical solution to \eqref{eq:sdefull} by the Euler-Maruyama scheme, for small time step $\delta > 0$ and $\xi_t \sim \mathcal{N}(0, \mathbb{I}_n)$, given by
\begin{equation}
    \label{eq:sdediscrete}
    X^\delta_{t+1} = X^\delta_{t} + \delta \, b(X^\delta_{t}) + \sqrt{\delta} \sigma(X^\delta_{t}) \xi_t, \quad X_0 = x \, .
\end{equation}
Analogously, the discrete-time version of \eqref{eq:sdesurrogate} is the following, where $\xi'_t \sim \mathcal{N}(0, \mathbb{I}_n)$: 
\begin{equation}
        \label{eq:sdesurrogate_discrete}
        \tilde{X}^\delta_{t+1} = \tilde{X}^\delta_{t} + \delta \, \tilde{b}_\theta(\tilde{X}^\delta_{t}) + \sqrt{\delta} \sigma(\tilde{X}^\delta_{t}) \xi'_t, \quad \tilde{X}_0 = x \, .
    \end{equation}
The numerical simulation of \eqref{eq:sdesurrogate_discrete} is used to generate a set of paths $\mathcal{J} = \{\tilde{X}^{\delta,(i)}:i=1,...,N_{\text{path}}\}$, where each path $\tilde{X}^{\delta(i)} = \{\tilde{X}^{\delta,(i)}_t \}_{t=1}^T$ is approximately sampled from $\tilde{\mathbb{P}}^x_{\theta,[0,T]}$. Given training data consisting of a collection of states $\mathcal{D} =\{x_j:j=1,...,N_{\text{samp}}\}$, the inexact goal-oriented loss is defined as 
\begin{equation}
    \label{eq:golossf_discrete}
    \begin{aligned}
    \hat{\mathcal{L}}_T^{\textup{GO}}(\theta; \mathcal{J}, \mathcal{D}) & = \frac{T}{2} \Big(\mathcal{M} + \frac{1}{N_{\text{path}}} \sum_{i=1}^{N_{\text{path}}} \phi_\tau^2 
    \big( \tilde{X}^{\delta,(i)} \big) \Big) \\
    & \quad \times \Bigg( \frac{1}{N_{\text{samp}}} \sum_{j=1}^{N_{\text{samp}}} \Big\| \sigma^+\big( b(x_j) - \tilde{b}_\theta(x_j) \big) \Big\|^2_2 \Bigg) \, .
    \end{aligned}
    \end{equation}
The inexact loss \eqref{eq:golossf_discrete} corresponds to an empirical estimator of either the forward or reverse goal-oriented loss depending on the empirical distribution of $\mathcal{D}$. When elements of $\mathcal{D}$ are approximately distributed according to $\pi$, \eqref{eq:golossf_discrete} is an estimator of $\mathcal{L}_T^\text{GO}$; when the elements are approximately distributed according to $\tilde{\pi}_\theta$, \eqref{eq:golossf_discrete} is an estimator of $\reverse{\mathcal{L}}_T^\text{GO}$. Empirical estimators of the loss functions reviewed in \Cref{sec:review} can be derived in a similar fashion. We identify three sources of error between a loss function for learning SDEs and its empirical estimator.

\textit{Finite-sample error.} Expectations within the loss function are estimated using Monte Carlo, where the variance of the estimator of each expectation scales with the sample size $N$ by $\mathcal{O}(N^{-1})$. In practice, the number of drift evaluations $N_\text{samp}$ may be fixed based on an existing dataset, and the number of paths $N_\text{path}$ which can be reasonably generated from the simulation of \eqref{eq:sdesurrogate_discrete} is limited by computational constraints, leading to greater estimator variance. 

\textit{Discretization error}. Due to time discretization, the Markov process $\tilde{X}^\delta$ from \eqref{eq:sdesurrogate_discrete} does not share the same law as the process $\tilde{X}$ from \eqref{eq:sdesurrogate}. Prior works have quantified the bias in the law of discretized SDEs under various assumptions on the drift and diffusion coefficients \cite{Zhang2016, Shao2018, Yu2024}, as well as the bias in the invariant distribution under strong log-concave conditions \cite{Dalalyan2017} and Poincar\'e or log-Sobolev conditions \cite{Vempala2019} on the invariant distribution.

\textit{Distribution error.} In general, the dataset $\mathcal{D}$ is expected to be an i.i.d. sample from the invariant distribution $\pi$ or $\tilde{\pi}_\theta$. In practice, these distributions are intractable to sample analytically, and Markov chain Monte Carlo (MCMC) methods may yield biased samples due to metastable effects -- where the path remains confined to one high probability region \cite{Roberts1996, Henin2022} -- and limitations to the simulation timescale, imposed by both the cost of drift evaluations and time step constraints for numerical stability \cite{Henin2022}. To address metastability in molecular dynamics, various techniques have been proposed, including biasing force methods \cite{Yang2019, Henin2022} and model order reduction via collective variables \cite{Mori2020, Frohlking2024}. Distribution error also occurs when the dataset $\mathcal{D}$ is assembled from heterogeneous sources; for example, training data for machine learning force fields are commonly a combination of short \textit{ab initio} simulations with different initial conditions or open-access quantum chemistry databases such as \cite{Jain2013, Ramakrishnan2014}. 

Since the errors from finite sample sizes and time discretization of the SDE are well understood, we focus the numerical experiment in \Cref{ex:mb} on studying the impact of distribution error on model learning outcomes.
\section{Numerical studies}
\label{sec:numerics}

In the following numerical experiments, we evaluate the performance of goal-oriented learning for overdamped Langevin systems of the form
\begin{equation}
    \label{eq:langevin}
    \textup{d}X_t = -\nabla_x V(X_t) \, \textup{d}t + \sqrt{2 \beta^{-1}} \textup{d}W_t, \qquad X_0 = x \, ,
\end{equation}
described by a confining potential function $V: \mathbb{R}^m \to \mathbb{R}$ \cite[Def.~4.2]{Pavliotis2014} and an inverse temperature parameter $\beta \in \mathbb{R}$. The invariant distribution of \eqref{eq:langevin} is the Gibbs-Boltzmann distribution $\pi\propto\exp(-\beta V)$ \cite{Pavliotis2014}. A family of model potentials parameterized by $\Theta \subseteq \mathbb{R}^d$,
\begin{align*}
    \mathcal{V} = \{ V_\theta: \mathbb{R}^m \to \mathbb{R} \ | \ &  \theta \in \Theta, 
    V_\theta \text{ is confining, $V_\theta$ and $\nabla_x V_\theta$ are $C^1$ in $\theta$}, \\
    & V_\theta \text{ admits invariant distribution } \pi_\theta \propto \exp(-\beta V_\theta) \} \, ,
\end{align*}
is specified such that SDEs based on $V$ and any model $V_\theta \in \mathcal{V}$ meet \Cref{cond:driftdiff,cond:elliptic}. The corresponding class of drift functions is  $\mathcal{B} = \{b_\theta=-\nabla_x V_\theta \ | \ V_\theta \in \mathcal{V}\}$.

While \Cref{cond:driftdiff,cond:elliptic} can also hold for underdamped Langevin systems---see \Cref{remark_underdamped_Langevin_systems}---we focus on overdamped Langevin systems in this work, for two reasons. First, underdamped Langevin systems sometimes require a careful choice of numerical integration schemes. By choosing overdamped Langevin systems, we can use simple and commonly used schemes, and thereby focus our analysis of computational performance on the methodology that we propose in \Cref{sec:golearning}. Second, overdamped Langevin systems are used as models in more scientific fields than underdamped Langevin systems, and offer a sufficiently rich class of systems to study.

\begin{remark}
The surrogate model learning examples in \Cref{ex:gmm,ex:mb} only partially satisfy conditions specified in \Cref{sec:gradloss} for the gradient formula in \Cref{prop:gradgoloss}. The examples consider mean first transition time observables over a time interval $[0,T]$, so that the stopping time is bounded and meets \Cref{cond:observable}. The partial derivatives $\partial_{\theta_i} (-\nabla_x V_\theta)$ are not globally bounded, as assumed in \Cref{cond:differentiable}; in practice, however, the confining nature of the potential ensures that the simulated paths remain within a subset of the state space where the partial derivatives are bounded. The Gaussian mixture potential of \Cref{ex:gmm} and neural network potential of \Cref{ex:mb} may not have the appropriate Lebesgue-integrable dominating functions required in \Cref{cond:domfcn}. It should be noted that \Cref{cond:differentiable,cond:observable,cond:domfcn} are sufficient but not necessary conditions for \Cref{prop:gradgoloss}, and the numerical results show that the gradient formulas are effective in minimizing the goal-oriented loss over gradient descent under the partial conditions.
\end{remark}

\subsection{Asymmetric double well potential}
\label{ex:dw}

The first example compares the goal-oriented error bound of \Cref{lem:2momentbd} with other information divergences. Consider the double well potential model on $\mathbb{R}$ of the following form, adopted from \cite{Opper2017}:
\begin{equation}
    \label{eq:dw}
    V_\theta(x) = x^4 - 2x^2 + \theta \ \Big( \frac{1}{3} x^3 + x^2 + x \Big) \, .
\end{equation}
The parameter $\theta \in [0,1]$ controls the degree of asymmetry between the two wells. For this example, we let $\beta=1$ and $x=-1$ in \eqref{eq:langevin}. The parametric variation in $V_\theta$ and $\pi_\theta \propto \exp(-\beta V_\theta)$ is illustrated in the left and center panels of \Cref{fig:doublewell}.
    \begin{figure}[]
    \centering
    \includegraphics[width=\textwidth]{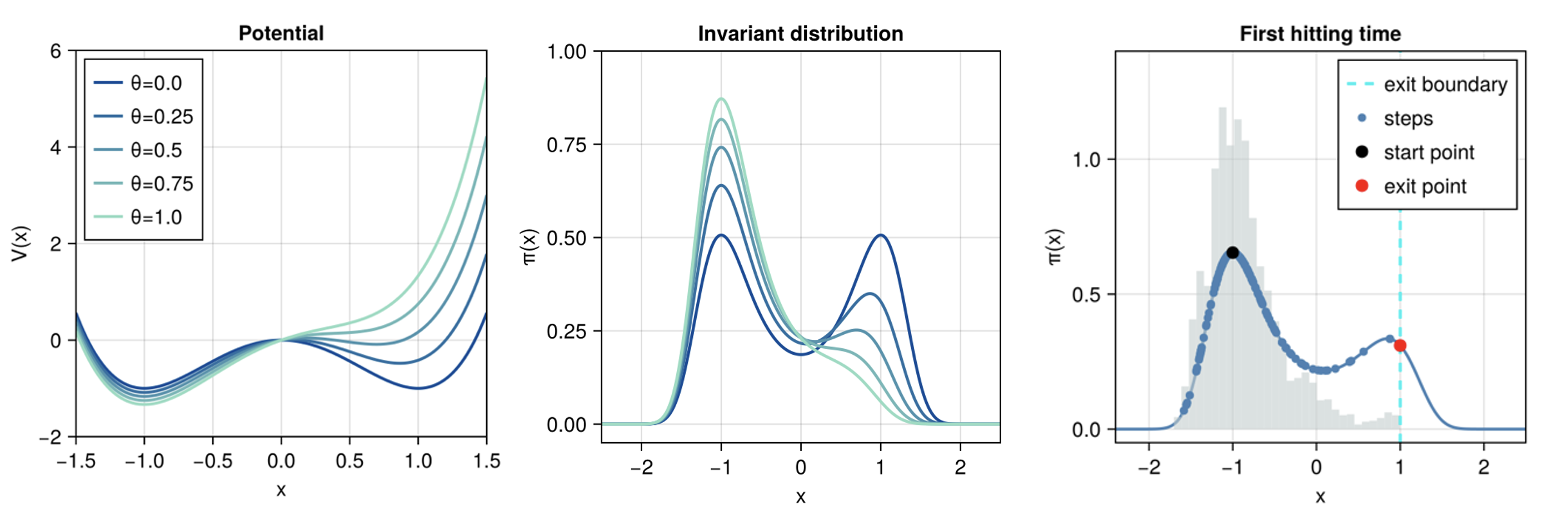}
    \vspace{-5mm}
    \caption{Instances of the double well potential in \eqref{eq:dw} (left) and the associated invariant distribution (center) as it varies with $\theta$. The first exit time from the domain $(\infty, 1]$ with initial condition $x=-1$ and the histogram of states visited by a single sample path (right).}
    \label{fig:doublewell}
    \end{figure}
    
We consider the well-specified setting where the reference potential is of the same form as the model, corresponding to $\theta^* = 0.5$. The path-space observable is the mean first exit time \eqref{eq:firstexittime} from the domain $B = (-\infty, 1]$, as illustrated in the right panel of \Cref{fig:doublewell}. We use a finite element approximation of the solution to the Feynman-Kac PDE in \eqref{eq:feynmankac} to obtain moments of the first exit time, as given in \eqref{eq:fkmoments}. Note that the first exit time is unbounded, meaning that it does not satisfy conditions for the bound in \eqref{eq:esssupbd} to hold. However, the SDE meets conditions for the first exit time to have finite second moments under $\mathbb{P}^x, \tilde{\mathbb{P}}^x_\theta$ and thus \Cref{lem:2momentbd} applies; see \Cref{rmk:hittime}. 

In \Cref{fig:dw_error_bounds}, we visualize how various quantities vary with respect to the parameter of the surrogate potential: the forward and reverse KL divergence between the invariant measures of the reference and surrogate SDEs, the forward and reverse KL divergence between their path measures as given in \Cref{def:path_kl}, the forward and reverse goal-oriented (GO) error bound from \eqref{eq:2momentbd}, and the absolute error in the observable $|\mathbb{E}^x[\tau] - \tilde{\mathbb{E}}^x_\theta[\tau]|$. In the well-specified setting, all curves attain their minimum value of zero at the reference parameter, $\theta^*=0.5$. In \Cref{fig:dw_error_bounds}, the KL divergence terms fall below the error in the observable for all values of $\theta$, indicating that none of these quantities provide an upper bound for the error. Meanwhile, the GO error bound satisfies an upper bound across the parameter domain, which is consistent with the theoretical result in \Cref{lem:2momentbd}.
    \begin{figure}[ht]
    \centering
    \includegraphics[width=0.45\textwidth]{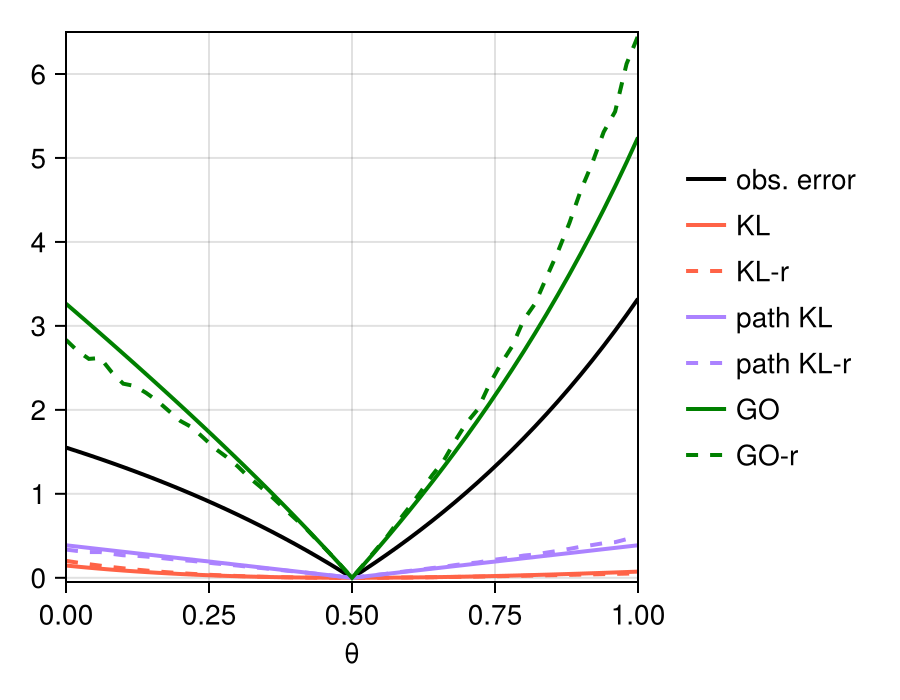}
    \caption{The absolute error in the mean first exit time (black) of the double well system in \eqref{eq:dw} compared to various forward (solid) and reverse (dashed) quantities: the KL divergence in invariant measures (KL and KL-r), the KL divergence in path measures (path KL and path KL-r), and the goal-oriented error bound (GO and GO-r). The reference model is defined by $\theta^*=0.5$.}
    \label{fig:dw_error_bounds}
    \end{figure}

\subsection{Confined Gaussian mixture model}
\label{ex:gmm}

Now consider a double well potential that is defined not as a polynomial function as in \Cref{ex:dw}, but as a mixture of two Gaussian distributions. This form serves as an example in which the parameterization of the model is nonlinear, multidimensional, and localizes features of the potential such as the width and location of energy wells. The potential is the Gaussian mixture model
\begin{equation}
    \label{eq:gmm}
    V_\theta(x) = \sum_{i=1}^2 w_i \exp \Bigg( - \frac{(x - c_i)^2}{2 \nu_i^2} \Bigg) + A(x-\bar{c})^4 \, ,
\end{equation}
where $\theta \coloneqq ( w_1,w_2,c_1,c_2,\log(\nu_1), \log(\nu_2) ) \in \mathbb{R}^6$ consist of the centers $c_i$, standard deviations $\nu_i$, and weights $w_i$ of the Gaussian mixture components $i=1,2$. We let $\bar{c}$ be the arithmetic average of the centers and $A =0.2$. The quartic polynomial in \eqref{eq:gmm} ensures that the potential is sufficiently confining for the invariant distribution to exist. The reference and surrogate potentials belong to the same potential family defined by \eqref{eq:gmm}, with $\theta^* = (1,1,-1,1,\log(0.2),\log(0.2))$. The path-space observable is the mean truncated first exit time in \eqref{eq:firstexittime} with $B=(-\infty,1]$, defined as $\tau(Y) \wedge T$. We set $T=10^4$ to be sufficiently large relative to the reference value of the mean first exit time, $\mathbb{E}^x[\tau]=1.853 \times 10^1$. 

In this example, we assess loss functions over iterations of gradient descent according to the error in the observable. We highlight that the true error will in general not be available in practical applications, as the reference value of the observable is typically unknown. The model is initialized with $\theta_0= ( 1,1,-1,-1,\log(0.75),\log(0.75) )$, shown in panel (a) of \Cref{fig:gmm_init_loss}. We apply the adaptive gradient (AdaGrad) algorithm \cite{Duchi2011} with a learning rate $\gamma=0.2$ to minimize the following four loss functions, 
\medskip
\begin{itemize}[leftmargin=*]
    \item Forward RER: \tabto{3.5cm} $\mathcal{L}^\text{RER}(\theta) \coloneqq \beta \mathbb{E}_{\pi_{\theta^*}}[ || \nabla_x V_{\theta^*} - \nabla_x V_\theta ||^2_2 ]$
    \smallskip
    \item Reverse RER: \tabto{3.5cm} $\reverse{\mathcal{L}}^\text{RER}(\theta) \coloneqq \beta \mathbb{E}_{\pi_\theta}[ || \nabla_x V_{\theta^*} - \nabla_x V_\theta ||^2_2 ]$
    \smallskip
    \item Forward GO loss: \tabto{3.5cm} $\mathcal{L}_T^{\textup{GO}}(\theta) \coloneqq T \beta \big( \mathcal{M} + \mathbb{E}^x_{\theta,[0,T]}[\phi^2_T]\big) \mathbb{E}_{\pi_{\theta^*}} \big[ ||\nabla_x V_{\theta^*} - \nabla_x V_\theta ||^2_2 \big]$
    \smallskip
    \item Reverse GO loss: \tabto{3.5cm} $\reverse{\mathcal{L}}_T^{\textup{GO}}(\theta) \coloneqq T \beta \big( \mathcal{M} + \mathbb{E}^x_{\theta,[0,T]}[\phi^2_T]\big) \mathbb{E}_{\pi_\theta} \big[ ||\nabla_x V_{\theta^*} - \nabla_x V_\theta ||^2_2 \big]$
\end{itemize}
\medskip
where $\pi_{\theta^*}, \pi_\theta$ denote the invariant distributions and $V_{\theta^*}, V_{\theta}$ denote the potential energy functions corresponding to the reference and surrogate SDEs, respectively. For purposes of comparison, $\mathcal{M}$ is set to the ground truth second moment of the observable. Since the AdaGrad algorithm adaptively normalizes gradients of the loss function, the parameter update at each optimization step is of the same order of magnitude for the four loss functions.

Empirical estimators of all the forward losses use a fixed dataset $\mathcal{D}$ of $N_\text{samp}=10^3$ samples from $\pi_{\theta^*}$, while empirical estimators of the reverse losses use $N_\text{samp}=10^3$ samples drawn from $\pi_\theta$ at each iterate of $\theta$. Sampling is performed via Langevin Monte Carlo using a $10^6$-step path of the surrogate SDE \eqref{eq:sdesurrogate_discrete} with step size $\delta=10^{-3}$, from which a uniformly random subset is selected in order to reduce temporal correlation of the data. Empirical estimators of the gradient of the GO losses use $N_\text{path}=10^3$ Monte Carlo path samples to compute the second moment of the first exit time under $\tilde{\mathbb{P}}_\theta$, as well as its Fr\'echet derivative using the formula \eqref{eq:2momentgrad}. 

    \begin{figure}[ht]
    \centering
    \includegraphics[width=\textwidth]{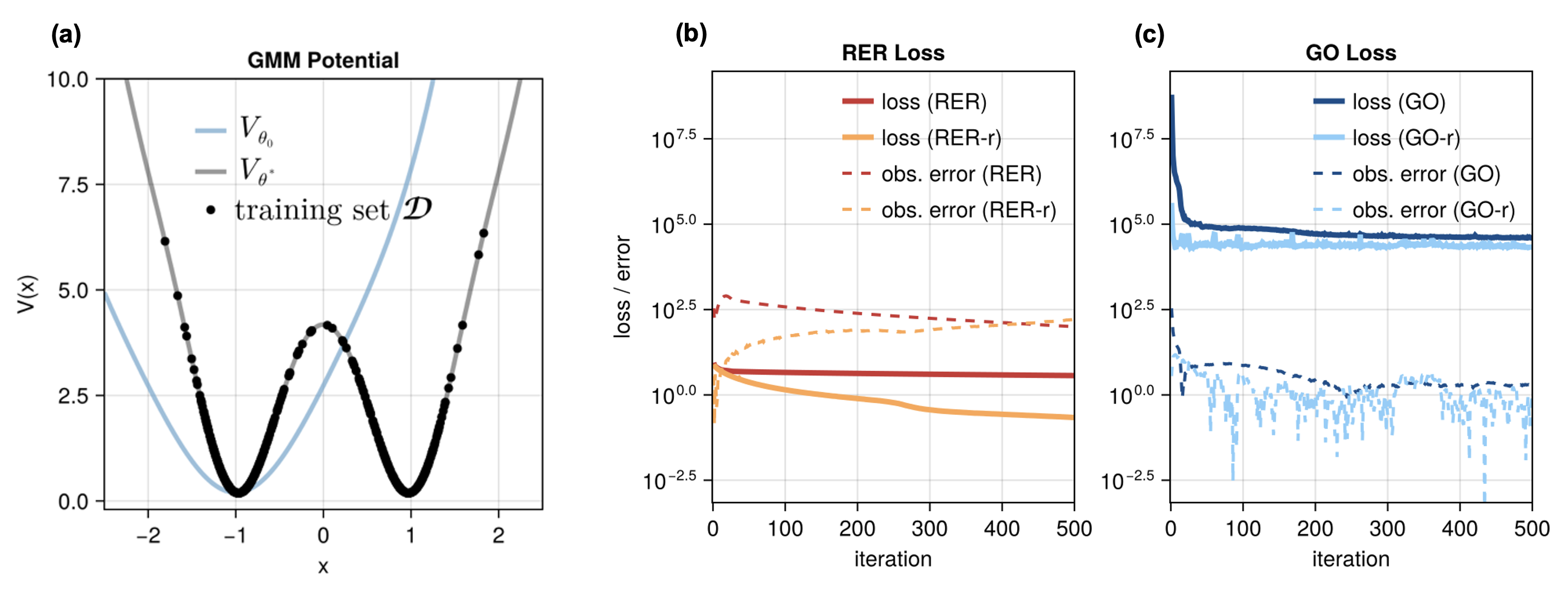}
    \vspace{-5mm}
    \caption{Reference and initial potentials, $V_{\theta^*}$ and $V_{\theta_0}$, of the form \eqref{eq:gmm} and dataset $\mathcal{D}$ from the reference process (a). Error in the observable (dashed) compared to the loss (solid) over 500 learning iterations with four objectives: the forward relative entropy rate (RER) and reverse relative entropy rate (RER-r) in (b), and the forward goal-oriented loss (GO) and reverse goal-oriented loss (GO-r) in (c).}
    \label{fig:gmm_init_loss}
    \end{figure}

While gradient descent is performed with inexact gradients, the underlying value of the loss and the error in the observable are computed deterministically in order to remove the effect of finite-sample or time-discretization error in the comparison of the loss functions. Expectations over invariant measures are computed with 200 Gauss-Legendre quadrature points, while moments of the first exit time are computed using a finite-element approximation of the solution to the Feynman-Kac PDE in \eqref{eq:feynmankac}.

Panels (b) and (c) of \Cref{fig:gmm_init_loss} compare each of the four losses over 500 iterations of gradient descent, along with the error in the observable at each iterate. The values of the GO losses are greater than their corresponding error in the observable for all iterations, providing numerical evidence of the upper bound relationship stated in \Cref{prop:golossineq}. In contrast, the RER losses fall below the error in the observable for nearly all iterations. Training with the GO losses decreases the error in the observable, with the reverse GO loss achieving the lowest error of all the loss functions.

For additional insight into the learning dynamics, one can track the evolution of components of the loss gradient over learning iterations. In the top row of \Cref{fig:gmm_loss_comp}, we decompose the GO loss into two terms which vary with $\theta$: $(\mathcal{M} + \tilde{\mathbb{E}}^x_{\theta,[0,T]}[(\tau \wedge T)^2])$ and the relative entropy rate (RER). For brevity, $\mathcal{H}$ is used as a general notation of RER, which can refer either to the forward case $\mathcal{H}(\mathbb{P}^x||\tilde{\mathbb{P}}^x_\theta)$ or reverse case $\mathcal{H}(\tilde{\mathbb{P}}^x_\theta || \mathbb{P}^x)$. In the forward GO loss, there is a sharp decrease in $\tilde{\mathbb{E}}^x_{\theta,[0,T]}[(\tau \wedge T)^2]$ and little change in the RER, suggesting that the second moment of the observable plays a more important role than the RER in the minimization of the total loss. The RER term of the reverse GO loss is more sensitive to changes in $\theta$, decreasing rapidly at early iterations. 

Similarly, we track the $\ell_2$ norm of the two components of the loss gradient in \Cref{prop:gradgoloss}, $G_1 \coloneqq ||\mathcal{H} \cdot \nabla_\theta \tilde{\mathbb{E}}^x_{\theta,[0,T]}[(\tau \wedge T)^2] ||_2$ and $G_2 \coloneqq ||(\mathcal{M} + \tilde{\mathbb{E}}^x_{\theta,[0,T]}[(\tau \wedge T)^2]) \nabla_\theta \mathcal{H}||_2$. In the bottom row of \Cref{fig:gmm_loss_comp}, we observe that $G_1$ is consistently larger than $G_2$ in forward GO learning, suggesting that $\nabla_\theta \tilde{\mathbb{E}}^x_{\theta,[0,T]}[(\tau \wedge T)^2]$ has a greater effect in gradient descent compared to $\nabla_\theta \mathcal{H}$. In contrast, $G_1$ and $G_2$ are of comparable magnitude in reverse GO learning, indicating that the two components have relatively equal contributions in gradient descent. The greater dependence on $\mathcal{H}$ in the reverse loss is attributed to the fact that reverse RER, which is an expectation with respect to $\tilde{\pi}_\theta$, has greater sensitivity to $\theta$ compared to the forward RER, where the expectation is with respect to a distribution independent of $\theta$. These trends suggest that the second moment of the observable plays a role of equal or greater significance to the RER in steering the learning dynamics.
\begin{figure}[ht]
    \centering
    \vspace{-2.5mm}
    \includegraphics[width=0.85\textwidth]{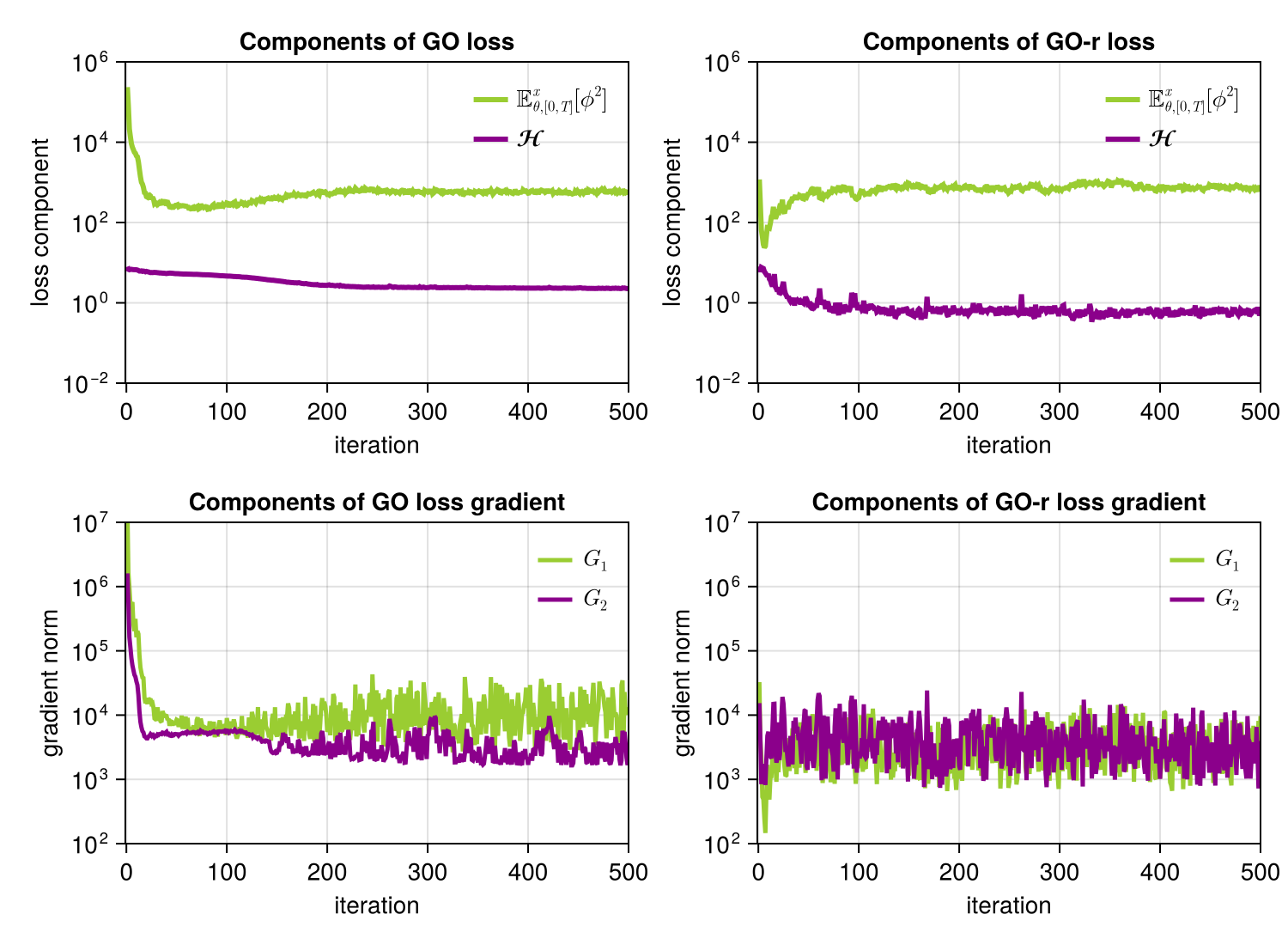}
    \caption{Components of the GO loss and its gradient for the double well system of \eqref{eq:gmm}. Top row: the second moment of the observable $\tilde{\mathbb{E}}^x_{\theta,[0,T]}[(\tau \wedge T)^2]$ and the relative entropy rate $\mathcal{H}$ over gradient descent iterations with the forward GO loss (left) and reverse GO loss (right). Bottom row: the $\ell_2$ norm of terms $G_1$ and $G_2$ of the loss gradient over gradient descent with the forward GO loss (left) and reverse GO loss (right). }
    \label{fig:gmm_loss_comp}
    \vspace{-2.5mm}
\end{figure}

\Cref{fig:gmm_stats} shows the mean and variance of the first exit time computed using iterations of the model generated by the gradient descent scheme. Both GO losses outperform their RER counterparts in terms of capturing the mean and variance of the observable; the observable estimated by learning with the reverse GO loss most quickly converges and oscillates about its reference value. These results suggest that models learned with the GO losses exhibit stronger control over statistics of the observable compared to those with the RER losses.

    \begin{figure}[ht]
    \centering
    \includegraphics[width=0.82\textwidth]{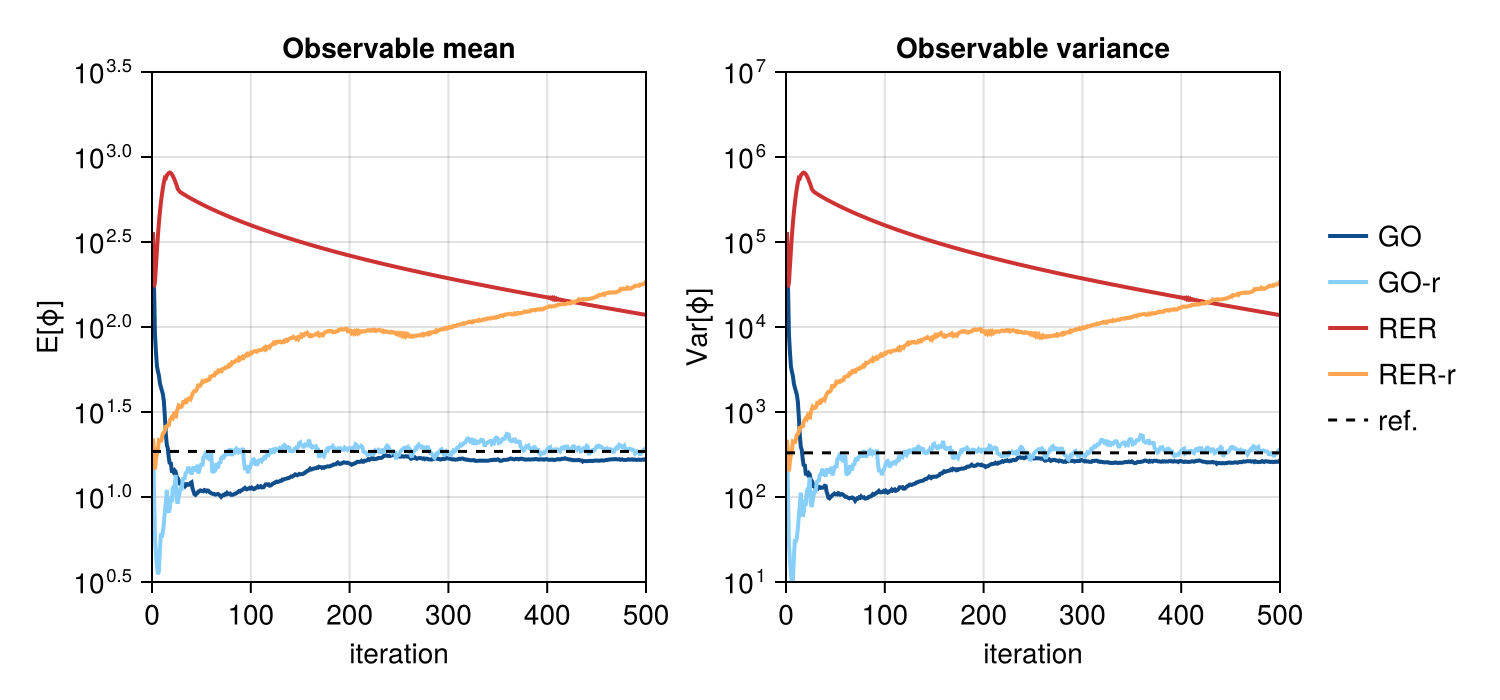}
    \caption{The mean (left) and variance (right) of the first exit time associated with the potential model of the form in \eqref{eq:gmm} over 500 gradient descent iterations with the four loss functions. The reference statistics of the first exit time are shown by the dashed black line.}
    \label{fig:gmm_stats}
    \end{figure}

\subsection{M\"uller-Brown potential}
\label{ex:mb}

In this example, we compare the robustness of loss functions to the level of distribution error in the training data $\mathcal{D}$, as discussed in \Cref{sec:inexactloss}. The reference system is defined by the M\"uller-Brown potential \cite{MullerBrown1979}, illustrated in the left panel of \Cref{fig:mb_hittime}, which is a widely used benchmark system in molecular dynamics. The surrogate system is defined by a feedforward neural network of the potential energy with a single 20-dimensional hidden layer and Tanh activation functions plus a confining quartic term, yielding trainable parameters $\theta \in \mathbb{R}^{501}$. The drift function, equal to the negative gradient of the potential function, is computed using autodifferentiation. Note that unlike the previous examples, the reference and surrogate drift functions belong to different model classes. The path-space observable studied is $\mathbb{E}[\tau \wedge T]$, the mean truncated first hitting time of a path initialized at $x=(-0.55, 0.45)$ to an elliptical domain $B$ enclosing an intermediate energy well of the M\"uller-Brown potential, where $\tau(Y)= \inf \{ t \geq 0: Y_t \in B\}$ and $T=10^4$.
    \begin{figure}[]
    \centering
    \includegraphics[width=\textwidth]{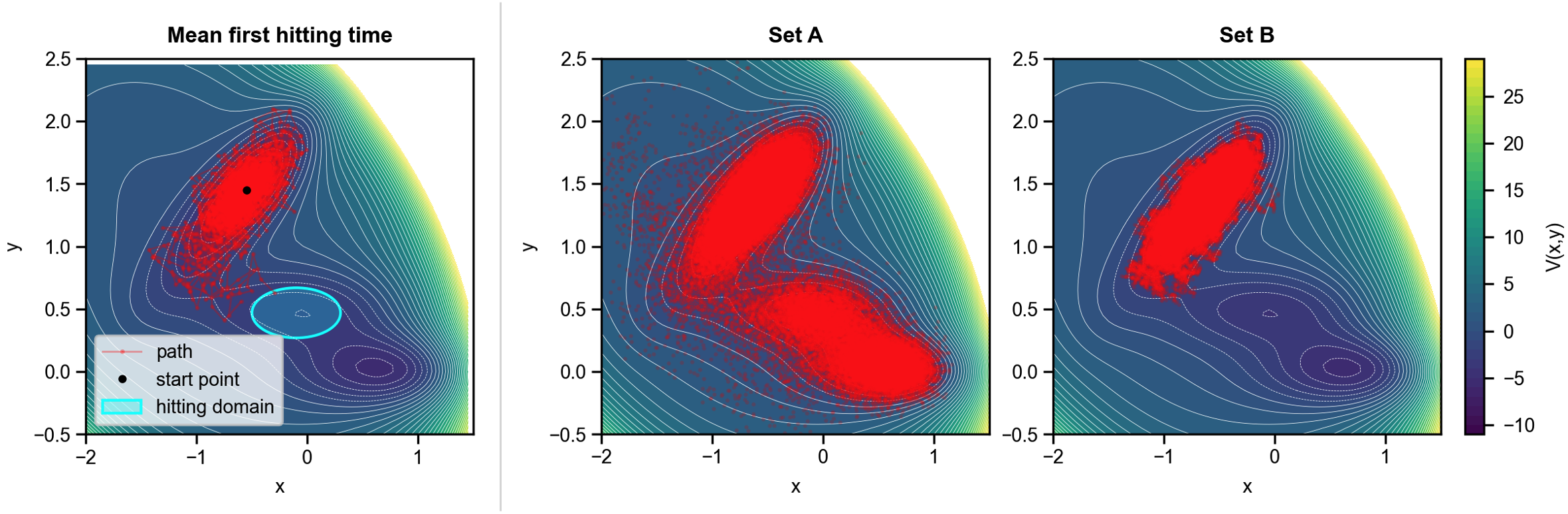}
    \vspace{-5mm}
    \caption{Initial point, hitting domain, and sample hitting path corresponding to the mean first hitting time with the M\"uller-Brown potential (left). Sets A (center) and B (right) from which the training datasets $\mathcal{D}_A$ and $\mathcal{D}_B$ are drawn.}
    \label{fig:mb_hittime}
    \end{figure}
    
To evaluate the effect of distribution error in forward losses, two cases are studied: one in which the distribution error is small, with data distributed approximately i.i.d. to the invariant distribution of the reference potential (from Set A of \Cref{fig:mb_hittime}), and one in which distribution error is high, with data concentrated in one metastable region of the reference potential (from Set B of \Cref{fig:mb_hittime}). Over 10 trials, training sets $\mathcal{D}_A$ and $\mathcal{D}_B$, each with $N_\text{samp}=5 \times 10^3$ points, are drawn as uniform random subsets of Sets A and B, respectively, each of which consist of a total of $10^5$ points. The datasets are then used in evaluating three forward losses:
\medskip
\begin{itemize}[leftmargin=*]
    \item Mean squared error in energy: \tabto{5cm} $\mathcal{L}^{\text{E}}(\theta) \coloneqq \mathbb{E}_{\pi}[| V - \tilde{V}_\theta |^2]$ \smallskip
    \item Mean squared error in forces: \tabto{4.98cm} $\mathcal{L}^{\text{F}}(\theta) \coloneqq \mathbb{E}_\pi[|| \nabla_x V - \nabla_x \tilde{V}_\theta ||^2_2]$ %
    \item Forward GO loss \eqref{eq:golossf}. 
\end{itemize}
\medskip
Mean squared error losses are a common choice of objective function for learning machine learning interatomic potentials in molecular dynamics \cite{Batzner2022, Musaelian2023}. Note that the mean squared error in forces is proportional to the RER loss used in \cite{Harmandaris2016}. In practice, we implement the energy matching (EM) loss $\mathcal{L}^\text{EM}$ from \eqref{eq:emloss} and the force matching (FM) loss $\mathcal{L}^\text{FM}$ from \eqref{eq:fmloss} as empirical estimators of the mean squared error in the potential energy and forces, respectively, and the inexact loss \eqref{eq:golossf_discrete} as an empirical estimator of the forward GO loss. In the inexact GO loss, $N_\text{path} = 100$ Monte Carlo path samples with $\delta=10^{-2}$ are used to calculate the second moment of the observable. Each of the losses are minimized over $10^3$ training epochs using AdaGrad with a learning rate $\gamma=10^{-2}$ and the same initial model parameters.

Panels (a) and (b) of \Cref{fig:mb_stats} show the median and interquartile range (IQR) of the predicted mean truncated first hitting time over 10 trials with each of the two datasets. With both datasets, the EM loss leads to consistently large deviations of the observable from the reference value. The FM loss leads to the largest variance in the predicted observable, indicated by the IQR. The GO loss leads to the closest agreement between the predicted observable and the ground truth, which remains consistent across the final training iterations. As shown in panel (c) of \Cref{fig:mb_stats}, the error in the observable from the EM and FM losses are higher at final training iterations when using data from Set B than from Set A, while the error from the GO loss remains less sensitive to the data distribution. These results suggest that the GO learning objective has the advantage of increased robustness to distribution error. 

    \begin{figure}[ht]
    \centering
    \includegraphics[width=\textwidth]{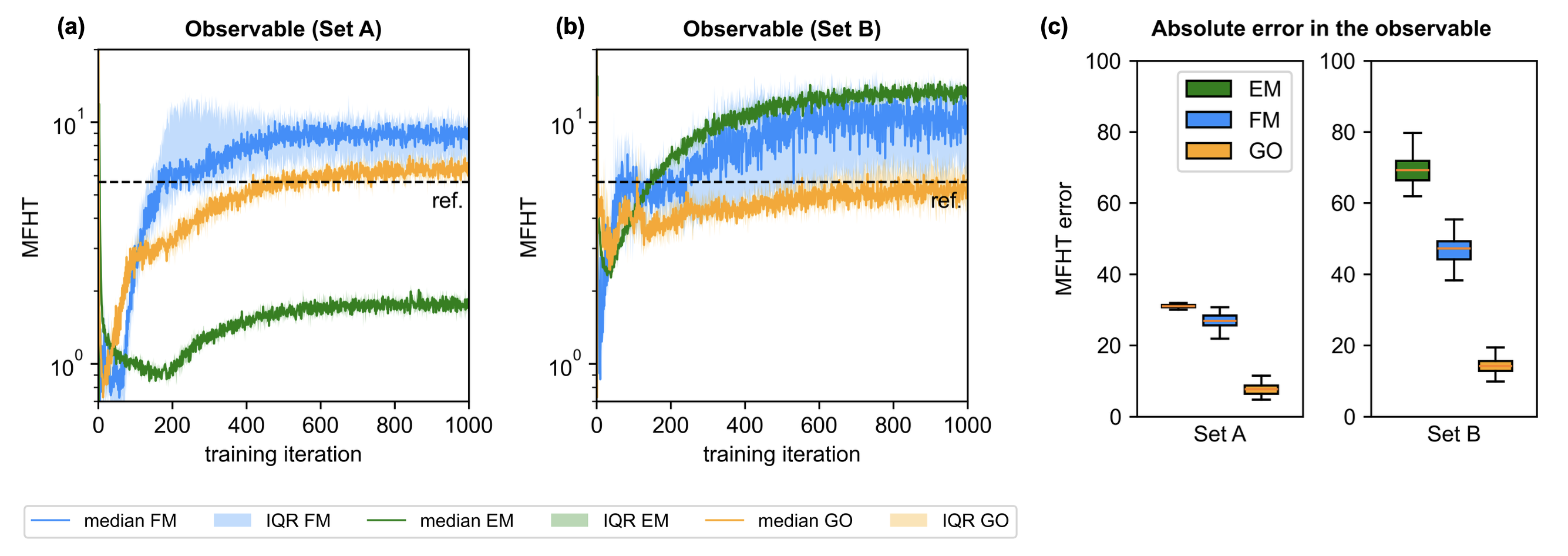}
    \vspace{-2mm}
    \caption{The mean first hitting time (MFHT) for the neural network model of the M\"uller-Brown potential, learned with $\mathcal{D}_A$(a) and with $\mathcal{D}_B$ (b), reported in terms of the median and interquartile range (IQR) over 10 trials. The reference value of the observable is shown in black. The distribution of the error in the observable across the final 100 iterations of all 10 trials, from the model learned with $\mathcal{D}_A$ and $\mathcal{D}_B$ (c). }
    \vspace{-5mm}
    \label{fig:mb_stats}
    \end{figure}

\section{Conclusions}
\label{sec:conclusions}

We develop a new form of goal-oriented loss function for learning surrogate models of SDEs which preserve the accuracy of path-space observables relative to an expensive reference process. The goal-oriented losses of \Cref{prop:golossineq} are derived from \Cref{lem:2momentbd}, an information inequality which bounds the absolute error in the path-space observable. Treating this bound as a learning objective yields a proxy regression problem, where minimizers of the objective control error in the observable. Compared to similar goal-oriented approaches \cite{Dupuis2016, Harmandaris2016, Branicki2021, Branicki2023}, our error bound applies symmetrically to a broader class of path-space observables, discussed in \Cref{rmk:comparebd}, while balancing accuracy -- i.e. tightness of the upper bound -- and computational simplicity to suit practical applications. In particular, we have a closed-form gradient of the goal-oriented loss in \Cref{prop:gradgoloss} which is straightforward to implement. The numerical experiments in \Cref{sec:numerics} suggest that in comparison to existing methods, our approach to goal-oriented learning can lead to models which reproduce first transition time statistics more accurately and are more robust to perturbations in the training data.

Goal-oriented learning is suitable in applications which require accuracy in quantifying a specific path-space observable rather than the full-field dynamics of the reference process. Consistent with the observable-based approaches reviewed in \Cref{sec:review}, the goal-oriented learning scheme weakens the notion of a regression solution to one where the surrogate SDE is identical to the reference SDE in terms of an expected path functional. As a result, there is a larger set of viable solutions that minimize the goal-oriented loss compared to other regression objectives, which can lead to faster convergence and improved robustness. However, a surrogate model developed via goal-oriented learning for one observable is not guaranteed to perform well in characterizing another observable. It is possible to incorporate multiple ``goals'' through a multiobjective optimization problem, in which the loss function is a combination of the goal-oriented loss for multiple observables.

A promising application of goal-oriented learning is to improve simulation-based estimation of reaction rates for kinetic Monte Carlo, a multiscale simulation technique for accelerating the computational characterization of chemical processes. An avenue for future work is to extend the goal-oriented scheme to underdamped Langevin systems, which characterize the dynamics of a broader class of molecular systems.

\bigskip

\appendix

\section{Proof of \Cref{prop:gradobs}}
\label{supp:gradobspf}

Recall that in the statement of \Cref{prop:gradobs}, $\tilde{\mathbb{E}}^x_{\theta}[\cdot]$ denotes the expectation with respect to $\tilde{\mathbb{P}}^x_\theta \in \mathcal{P}(C([0,\infty),\mathbb{R}^{m}))$, i.e., the law of the solution to the modified SDE \eqref{eq:sdeparametric} with drift $\tilde{b}_\theta\in\mathcal{B}$, where $\mathcal{B}$ is defined in \Cref{sec:goloss}.
Analogously, we shall write $\tilde{\mathbb{E}}_{\tilde{b}}[\phi_\tau]$ to denote the expectation of $\phi_\tau$ with respect to the law of \eqref{eq:sdeparametric} when $\tilde{b}_\theta$ in \eqref{eq:sdeparametric} is replaced by an arbitrary $\tilde{b}\in\mathcal{B}$.

To prove \Cref{prop:gradobs}, we first state \Cref{lem:Frechet} below.

\begin{lemma}[Fr\'echet derivative of path-space observables]
    \label{lem:Frechet}
    Assume Conditions \ref{cond:elliptic} and \ref{cond:observable} hold, and define the map $\mathcal{T}: \mathcal{B} \to \mathbb{R}$ to be $\tilde{b} \mapsto \tilde{\mathbb{E}}_{\tilde{b}}[\phi_\tau]$. Then the Fr\'echet derivative of $\mathcal{T}$ at $\tilde{b} \in \mathcal{B}$ in the direction of an arbitrary bounded, Borel-measurable function $v:\mathbb{R}^{m}\to\mathbb{R}^{m}$ is as follows, where $M^v_\tau$ is defined as in \eqref{eq:martingale}:
    \[
        \delta\mathcal{T}(\tilde{b}; v ) = \tilde{\mathbb{E}}_{\tilde{b}}[\phi_\tau M^v_\tau] \, .
    \]
    
\end{lemma}
\begin{proof}[Proof of \Cref{lem:Frechet}]
\Cref{cond:elliptic,cond:observable} imply Assumption 3.1 and Assumption 3.2 of \cite{Lie2021} respectively. The conclusion now follows by \cite[Lemma 3.6]{Lie2021}.
\end{proof}

\begin{proof}[Proof of \Cref{prop:gradobs}] 
 
Define the map $\tilde{b}: \Theta \to \mathcal{B}$ to be $\theta\mapsto \tilde{b}(\theta)\coloneqq \tilde{b}_\theta$, for $\tilde{b}_\theta$ as in \eqref{eq:sdeparametric}. Then for the map $\mathcal{T}$ in \Cref{lem:Frechet}, we can write the path-space observable $\tilde{\mu}: \Theta \to \mathbb{R}$ from \Cref{cond:observable} as the composition $\tilde{\mu} = \mathcal{T} \circ \tilde{b}$. By the chain rule, the Fr\'echet derivative of $\tilde{\mu}$ at $\theta \in \Theta \subseteq \mathbb{R}^d$ in the direction of the $i$th canonical basis vector, $e_i \in \mathbb{R}^d$, for $i=1,...,d$ is
        \begin{equation}
            \label{eq:pf_frechet}
            \delta\tilde{\mu}(\theta; e_i) = \delta(\mathcal{T} \circ \tilde{b})(\theta; e_i) = \delta\mathcal{T}(\tilde{b}; \delta \tilde{b}(\theta;e_i)) \, .
        \end{equation}
    Since \Cref{cond:differentiable} holds, the vector $\delta \tilde{b}(\theta;e_i)$ is exactly the derivative of $\tilde{b}$ evaluated at $\theta$ in the direction $e_i$:
        \begin{equation*}
            \delta \tilde{b}(\theta;e_i) = \nabla_\theta \tilde{b}(\theta) \cdot e_i = \partial_{\theta_i}\tilde{b}(\theta) \, .
        \end{equation*}
   Combining the preceding equation with \Cref{lem:Frechet} and \eqref{eq:pf_frechet}, we have
        \begin{equation}
            \delta\tilde{\mu}(\theta; e_i) = \tilde{\mathbb{E}}_{\tilde{b}}\big[ \phi_\tau M_\tau^{\partial_{\theta_i} \tilde{b}(\theta)} \big] \, .
        \end{equation}
    Since $\theta\in\mathbb{R}^{d}$ is on a normed space, it follows that the Fr\'echet derivative coincides with the Gateaux derivative \cite[Sec. 7.2, Prop. 2]{Luenberger1969}, meaning one can write
    \begin{align*}
            \delta\tilde{\mu}(\theta; e_i) = \lim_{\varepsilon \to 0} \bigg( \frac{\tilde{\mu}(\theta+\varepsilon e_i) - \tilde{\mu}(\theta)}{\varepsilon}\bigg) = \nabla_\theta\tilde{\mu}(\theta) \cdot e_i = \partial_{\theta_i} \tilde{\mu}(\theta) \, .
    \end{align*}

\end{proof}

\section{Proof of \Cref{prop:gradgoloss}}
\label{supp:proppf}

We start by stating and proving the formula for the gradient with respect to $\theta$ of an expectation of a function over a probability distribution on $\mathbb{R}^m$, where both the function and probability distribution depend on $\theta$. This formula is referred to as the \textit{score function gradient estimator} and sufficient conditions for the formula to hold are presented in \cite[Sec. 4.3.1]{Mohamed2020}.

\begin{lemma}[Score function gradient estimator]
\label{lem:scorefunc}
Let $\rho,g:\mathbb{R}^m\times\Theta\to\mathbb{R}$ be such that for every $\theta\in\Theta$, $\rho(\cdot,\theta)$ is a Lebesgue probability density and $g(\cdot,\theta)$ is $\rho(\cdot,\theta)$-integrable. Moreover, let $\rho(x, \theta) > 0$ and $\nabla_\theta \rho(x,\theta)$ and $\nabla_\theta g(x,\theta)$ be defined for almost every $x$. Assume that there exists a function $\psi \in L^1$ such that $\sup_{\theta \in \Theta} || \nabla_\theta \big( g(x, \theta) \rho(x, \theta) \big)||_1 \leq \psi(x)$ for almost every $x$. Then the gradient of $\mathbb{E}_{\rho}[g]$ with respect to $\theta$ is 
    \begin{equation}
        \label{eq:scorefuncest}
        \nabla_\theta \mathbb{E}_{\rho}[g] = \mathbb{E}_{\rho} [ g \nabla_\theta \log \rho + \nabla_\theta g] \, .
    \end{equation}
\end{lemma}

\begin{proof}
    The key to the derivation is the form of the score function, or gradient log probability, $\nabla_\theta \log \rho(x, \theta) = \frac{\nabla_\theta \rho(x, \theta)}{\rho(x, \theta)}$. 
    \begin{align*}
            \nabla_\theta \mathbb{E}_{x \sim \rho}[g(x, \theta)] & = \nabla_\theta \int g(x, \theta) \rho(x, \theta) \textup{d}x 
            \\
            & = \int \nabla_\theta \big[ g(x, \theta) \rho(x, \theta) \big] \textup{d}x
            \\
            & = \int g(x, \theta) \nabla_\theta \rho(x, \theta) \textup{d}x + \int \rho(x, \theta) \nabla_\theta g(x, \theta) \textup{d}x 
            \\
            & = \int g(x, \theta) \big[ \nabla_\theta \log \rho(x, \theta) \big] \rho(x, \theta) \textup{d}x + \int \rho(x, \theta) \nabla_\theta g(x, \theta) \textup{d}x 
            \\
            & = \mathbb{E}_{\rho(\cdot,\theta)} [ g(\cdot,\theta) \nabla_\theta \log \rho(\cdot,\theta)] + \mathbb{E}_{\rho(\cdot,\theta)} [\nabla_\theta g(\cdot,\theta)] \, . 
    \end{align*}

\end{proof}

The formula \eqref{eq:scorefuncest} holds when the interchange of the order of differentiation and integration in the second line of the proof is valid, which we establish using the existence of a dominating function $\psi \in L^1$ and the dominated convergence theorem. Other sufficient conditions are discussed in \cite{Lecuyer1995}. 

We have the following corollary to \Cref{lem:scorefunc} for the case where $\rho$ corresponds to a Gibbs-Boltzmann distribution. 

\begin{corollary}
    \label{corr:scorefunc}
    Let $\pi(x, \theta) =\exp(-\beta V(x, \theta))/Z(\theta)$ be a Gibbs-Boltzmann distribution for some $\theta$-differentiable potential $V: \mathbb{R}^m \times \Theta \to \mathbb{R}$, with normalizing constant $Z(\theta) = \allowbreak \int \exp(-\beta V(x, \theta)) \textup{d}x$. For $g:\mathbb{R}^m\times\Theta\to\mathbb{R}$, let $g(\cdot,\theta)$ be a $\pi(\cdot,\theta)$-integrable function for which there exists a dominating function $\psi \in L^1$ such that $\sup_{\theta \in \Theta} \allowbreak || \nabla_\theta \big( g(x, \theta) \pi(x, \theta) \big)||_1 \leq \psi(x)$ for almost every $x$. Then the gradient of $\mathbb{E}_{\pi}[g]$ with respect to $\theta$ is
    \begin{equation}
        \label{eq:scorefuncest_gibbs}
        \nabla_\theta \mathbb{E}_{\pi}[g] = \mathbb{E}_{\pi}\big[ \nabla_\theta g - \beta g (\nabla_\theta V - \mathbb{E}_{\pi}[\nabla_\theta V]) \big] \, .
    \end{equation}
\end{corollary}

\begin{proof}
    We have the following closed-form expression for the gradient log probability corresponding to $\pi$: 
    \begin{align*}
        \nabla_\theta \log \pi(x, \theta) & = \nabla_\theta \log \Big( \frac{\exp(-\beta V(x, \theta))}{Z(\theta)}\Big) = \nabla_\theta (-\beta V(x, \theta) - \log Z(\theta)) \\
        & = -\beta \nabla_\theta V(x, \theta) - \frac{1}{Z(\theta)} \nabla_\theta Z(\theta) \\
        & = -\beta \nabla_\theta V(x, \theta) - \frac{1}{Z(\theta)} \nabla_\theta \int \exp(-\beta V(x, \theta)) \textup{d}x \\
        & = -\beta \nabla_\theta V(x, \theta) - \frac{1}{Z(\theta)}  \int \nabla_\theta \big( \exp(-\beta V(x, \theta)) \big) \textup{d}x \\
        & = -\beta \nabla_\theta V(x, \theta) + \int \frac{\exp(-\beta V(x, \theta))}{Z(\theta)} \beta \nabla_\theta V(x, \theta)  \textup{d}x \\
        & = -\beta \big( \nabla_\theta V(x, \theta) - \mathbb{E}_{\pi(\cdot,\theta)} [\nabla_\theta V(\cdot,\theta)] \big) \, .
    \end{align*}
Substituting this expression in place of $\nabla_\theta \log \rho$ in \Cref{lem:scorefunc} recovers \eqref{eq:scorefuncest_gibbs}.
\end{proof}

\begin{proof}[Proof of \Cref{prop:gradgoloss}] The gradients of the forward and reverse goal-oriented losses in \eqref{eq:golossf} and \eqref{eq:golossr} are obtained in \eqref{eq:golossf_grad} and \eqref{eq:golossr_grad} by applying the chain rule.  Two gradients appear in the expressions: that of the relative entropy rate, $\nabla_\theta \mathcal{H}(\mathbb{P}^x || \tilde{\mathbb{P}}^x_\theta)$ or $\nabla_\theta \mathcal{H}(\tilde{\mathbb{P}}^x_\theta || \mathbb{P}^x)$, and that of the second moment of the observable, $\nabla_\theta \tilde{\mathbb{E}}^x_{\theta,[0,T]} [\phi_T^2]$.

The expression for the gradient of the forward relative entropy rate, $\nabla_\theta \mathcal{H}(\mathbb{P}^x || \tilde{\mathbb{P}}^x_\theta)$, is obtained in \eqref{eq:rerf_grad} as follows. Under \Cref{cond:elliptic,cond:domfcn}, $\sup_\theta \allowbreak \big| \big| \nabla_\theta \big( || \sigma^+(b - \tilde{b}_\theta) ||_2^2 \big) \big| \big|_1$ is dominated by an integrable function. Therefore, one can interchange the order of differentiation and integration to obtain the gradient of the forward relative entropy rate with respect to $\theta$ as
    \[
    \nabla_\theta \mathcal{H}(\mathbb{P}^x || \tilde{\mathbb{P}}^x_\theta) = \frac{1}{2} \mathbb{E}_\pi \Big[ \nabla_\theta \big( || \sigma^+(b - \tilde{b}_\theta) ||_2^2 \big) \Big] = \mathbb{E}_\pi \Big[ J_\theta^\top\Sigma^+ (\tilde{b}_\theta - b) \Big] \, ,
    \]
where the Jacobian $J_\theta$ is defined in \eqref{eq:jacobian}. 

Since \Cref{cond:domfcn} is satisfied, one can apply the score function gradient estimator formula from \Cref{lem:scorefunc} to obtain the gradient of the reverse relative entropy rate, $\nabla_\theta \mathcal{H}(\tilde{\mathbb{P}}^x_\theta || \mathbb{P}^x)$, in \eqref{eq:rerr_grad}:
\[
\begin{aligned}
    \nabla_\theta \mathcal{H}(\tilde{\mathbb{P}}^x_\theta || \mathbb{P}^x) & = 
    \frac{1}{2} \nabla_\theta \mathbb{E}_{\tilde{\pi}_\theta} \Big[ || \sigma^+(b - \tilde{b}_\theta) ||_2^2 \Big] \\
    & = \frac{1}{2} \mathbb{E}_{\tilde{\pi}_\theta} \Big[ || \sigma^+(b - \tilde{b}_\theta) ||_2^2 \ \nabla_\theta \log \tilde{\pi}_\theta + \nabla_\theta \big( || \sigma^+(b - \tilde{b}_\theta) ||_2^2 \big) \Big] \\
    & = \frac{1}{2} \mathbb{E}_{\tilde{\pi}_\theta} \Big[ || \sigma^+(b - \tilde{b}_\theta) ||_2^2 \ \nabla_\theta \log \tilde{\pi}_\theta + 2 J_\theta^\top \Sigma^+(\tilde{b}_\theta - b) \Big] \, .
    \end{aligned}
\]

The expression for the gradient of the second moment of the path-space observable, $\nabla_\theta \tilde{\mathbb{E}}^x_{\theta,[0,T]} [\phi_T^2]$, is obtained in \eqref{eq:2momentgrad} as follows. Since $\phi_T$ is defined on a bounded time interval $[0,T]$, the stopping time is bounded and meets \Cref{cond:observable}. Therefore, one can apply \Cref{prop:gradobs} to each of the partial derivatives to obtain

\[
    \begin{aligned}
    \nabla_\theta \tilde{\mathbb{E}}^x_{\theta,[0,T]} [\phi_T^2] & =
    \begin{bmatrix} 
    \partial_{\theta_1} \tilde{\mathbb{E}}^x_{\theta,[0,T]} [\phi_T^2] \\ \vdots \\ \partial_{\theta_d} \tilde{\mathbb{E}}^x_{\theta,[0,T]} [\phi_T^2]
    \end{bmatrix} = 
    \begin{bmatrix} 
        \tilde{\mathbb{E}}^x_{\theta,[0,T]} \big[ \phi_T^2 \int_0^T (\sigma^+ \partial_{\theta_1} \tilde{b}_\theta)^\top \textup{d}W_s \big] \\
        \vdots \\
        \tilde{\mathbb{E}}^x_{\theta,[0,T]} \big[ \phi_T^2 \int_0^T (\sigma^+ \partial_{\theta_d} \tilde{b}_\theta)^\top \textup{d}W_s \big] 
    \end{bmatrix} \\
    & = \tilde{\mathbb{E}}^x_{\theta,[0,T]}
    \begin{bmatrix}
        \phi_T^2 \begin{bmatrix}
            \int_0^T (\sigma^+ \partial_{\theta_1} \tilde{b}_\theta)^\top \textup{d}W_s \\
            \vdots \\
            \int_0^T (\sigma^+ \partial_{\theta_d} \tilde{b}_\theta)^\top \textup{d}W_s
        \end{bmatrix}
    \end{bmatrix} \\
    & = \tilde{\mathbb{E}}^x_{\theta,[0,T]} \Big[ \phi_T^2 \int_0^T (\sigma^+J_\theta)^\top \textup{d}W_s \Big] = \tilde{\mathbb{E}}^x_{\theta,[0,T]}[\phi_T^2 M_T^{J_\theta}] \, ,
    \end{aligned} 
\]
where the $d$-dimensional local martingale $(M^{J_\theta}_t)_{t \geq 0}$ is defined in \eqref{eq:dmartingale}. 
\end{proof}

When the invariant distribution of the surrogate SDE corresponds to a Gibbs-Boltzmann distribution, e.g. $\tilde{\pi}_\theta \propto \exp(-\beta \tilde{V}_\theta)$ using the shorthand $\tilde{V}_\theta(\cdot) = \tilde{V}(\cdot, \theta)$, then the gradient log probability $\nabla_\theta \log \tilde{\pi}_\theta$ has an analytical form and the expression for the gradient of the reverse relative entropy rate in \eqref{eq:rerr_grad} can be further reduced to
\[
    \nabla_\theta \mathcal{H}(\tilde{\mathbb{P}}^x_\theta || \mathbb{P}^x) = \frac{1}{2} \mathbb{E}_{\tilde{\pi}_\theta} \Big[  - \beta || \sigma^+(b - \tilde{b}_\theta) ||_2^2 \ \big( \nabla_\theta \tilde{V}_\theta - \mathbb{E}_{\tilde{\pi}_\theta}[\nabla_\theta \tilde{V}_\theta] \big) + 2J_\theta^\top \Sigma^+(\tilde{b}_\theta - b) \Big]
\]
as a result of \Cref{corr:scorefunc}. 

\section{Goal-oriented learning algorithm}
\label{supp:alg}

In \Cref{alg:golossf,alg:golossr}, we present the pseudocode for goal-oriented learning with empirical estimators of the forward loss and reverse loss in \eqref{eq:golossf} and \eqref{eq:golossr}, respectively, using a gradient descent algorithm. The following notation is used:
\begin{table}[H]
\centering
\small
\begin{tabular}{ll}
\textit{Variable} & \textit{Description}                 \\ \hline
$\theta_0$        & initial parameter vector                                         \\
$K$        & total number of gradient descent iterations                      \\
$N_\text{samp}$        & number of sample states from $\tilde{\pi}_{\theta_k}$            \\
$\mathcal{D} \coloneqq \{x_j:j=1,...,N_\text{samp} \}$        & dataset of sample states                            \\
$N_\text{path}$        & number of sample paths from $\tilde{\mathbb{P}}^x_{\theta_k}$    \\
$\mathcal{J} \coloneqq \{\tilde{X}^{\delta,(i)}:i=1,...,N_\text{path} \}$        & dataset of sample paths                                          \\
$T$        & maximum time length of paths                                     \\
$X_0$        & initial condition of paths                                           \\
$\delta$        & time step                                                        \\
$\gamma$       & learning rate in gradient descent                                \\
$\varepsilon_\theta$       & tolerance for incremental change in $\theta$          \\
$\varepsilon_\mathcal{L}$       & tolerance for incremental change in $\hat{\mathcal{L}}^\textup{GO}_T$\\ \hline
\end{tabular}
\end{table}

The pseudocode refers to the following methods. \textsc{SampleInvariantDist} is a method for drawing $N_\text{samp}$ samples from the invariant distribution $\tilde{\pi}_\theta$ of the surrogate SDE in \eqref{eq:sdesurrogate}. One may implement Langevin Monte Carlo or another choice of Markov chain Monte Carlo (MCMC) scheme. \textsc{SampleHittingPaths} computes $N_\text{path}$ realizations of a time-discretized solution to the surrogate SDE; for instance, using the Euler Maruyama discretization scheme in \eqref{eq:sdesurrogate_discrete}. \textsc{GOLoss-f} and \textsc{GOLoss-r} are empirical estimators of the forward and reverse goal-oriented losses, \eqref{eq:golossf} and \eqref{eq:golossr}, respectively. Both estimators take the form of \eqref{eq:golossf_discrete}, differing only in the dataset $\mathcal{D}$ of sample states used to compute Monte Carlo averages. In the forward loss, $\mathcal{D}$ approximately represents the invariant distribution of the reference SDE $\pi$ and remains fixed across iterations, whereas in the reverse loss, $\mathcal{D}_k$ is sampled at every iteration from the invariant distribution of the surrogate SDE at the current parameter iterate $\theta_k$. Similarly, \textsc{GradientGOLoss-f} and \textsc{GradientGOLoss-r} are empirical estimators of the gradient formula for the forward and reverse goal-oriented losses in \eqref{eq:golossf_grad} and \eqref{eq:golossr_grad}, respectively.

The termination criterion in the gradient descent loop is based on the absolute change in the parameter or loss function, up to a terminal number of iterations $K$. The algorithms are presented using basic gradient descent for the sake of simplicity though other gradient-based optimization schemes can be implemented such as adaptive gradient descent \cite{Duchi2011}. 

\begin{algorithm}[H]
\caption{Goal-oriented learning (forward loss)}\label{alg:golossf}
\begin{algorithmic}
    \Require{$\theta_0, K, \mathcal{D}, N_\text{path}, T, X_0, \delta, \gamma, \varepsilon_\theta, \varepsilon_\mathcal{L}$} 
            
        \For{$k = 0, \dots, K-1$}
        
            \State $\mathcal{J}_k \gets$ \Call{SampleHittingPaths}{$\theta_k, N_\text{path},T, X_0, \delta$}

            \State $\hat{\mathcal{L}}_T^\text{GO}(\theta_k) \gets$ \Call{GOLoss-f}{$\theta_k, \mathcal{D}, \mathcal{J}_k$} 

            \State $\nabla_\theta \hat{\mathcal{L}}_T^\text{GO}(\theta_k) \gets$ \Call{GradientGOLoss-f}{$\theta_k, \mathcal{D}, \mathcal{J}_k$} 

            \State $\theta_{k+1} \gets \theta_k - \gamma \nabla_\theta \hat{\mathcal{L}}_T^\text{GO}(\theta_k)$

            \If{$\left\| \theta_{k+1} - \theta_k \right\| < \varepsilon_\theta$ \textbf{ or } $| \hat{\mathcal{L}}_T^\text{GO}(\theta_{k+1}) - \hat{\mathcal{L}}_T^\text{GO}(\theta_k) | < \varepsilon_\mathcal{L}$ \textbf{ or } $k+1=K$}

                \State $\hat{\theta} \gets \ \theta_{k+1}$
            \EndIf
        \EndFor
\end{algorithmic}
\end{algorithm}

\begin{algorithm}[H]
\caption{Goal-oriented learning (reverse loss)}\label{alg:golossr}
\begin{algorithmic}
    \Require{$\theta_0, K, N_\text{samp}, N_\text{path}, T, X_0, \delta, \gamma, \varepsilon_\theta, \varepsilon_\mathcal{L}$}
            
        \For{$k = 0, \dots, K-1$}

            \State $\mathcal{D}_k \gets$ \Call{SampleInvariantDist}{$\theta_k, N_\text{samp}, T, \delta$}
            
            \State $\mathcal{J}_k \gets$ \Call{SampleHittingPaths}{$\theta_k, N_\text{path},T, X_0, \delta$}

            \State $\hat{\reverse{\mathcal{L}}}_T^\text{GO}(\theta_k) \gets$ \Call{GOLoss-r}{$\theta_k, \mathcal{D}_k, \mathcal{J}_k$}

            \State $\nabla_\theta \hat{\reverse{\mathcal{L}}}_T^\text{GO}(\theta_k) \gets$ \Call{GradientGOLoss-r}{$\theta_k, \mathcal{D}_k, \mathcal{J}_k$}

            \State $\theta_{k+1} \gets \theta_k - \gamma \nabla_\theta \hat{\reverse{\mathcal{L}}}_T^\text{GO}(\theta_k)$

            \If{$\left\| \theta_{k+1} - \theta_k \right\| < \varepsilon_\theta$ \textbf{ or } $| \hat{\reverse{\mathcal{L}}}_T^\text{GO}(\theta_{k+1}) - \hat{\reverse{\mathcal{L}}}_T^\text{GO}(\theta_k) | < \varepsilon_\mathcal{L}$ \textbf{ or } $k+1=K$}

                \State $\hat{\reverse{\theta}} \gets \ \theta_{k+1}$
            \EndIf
        \EndFor
\end{algorithmic}
\end{algorithm}

\section*{Acknowledgments}
The authors would like to thank Olivier Zahm for helpful discussions on this work. JZ and YM are supported by the United States Department of Energy, National Nuclear Security Administration under Award Number DE-NA0003965. JZ and HCL are supported in part by the Deutsche Forschungsgemeinschaft (DFG) --- Project-ID \href{https://gepris.dfg.de/gepris/projekt/318763901}{318763901} --- SFB1294. 

\bibliographystyle{siamplain}
\bibliography{references}

@book{Pavliotis2014,
    author = {Pavliotis, G.},
    title = {{Stochastic Processes and Applications}},
    publisher = {Springer, Texts in Applied Mathematics},
    year = {2014},
    doi = {10.1007/978-1-4939-1323-7}
}

@book{Oksendal2013,
    author = {Oksendal, B.},
    title = {{Stochastic Differential Equations: An Introduction with Applications (6th ed.)}},
    publisher = {Springer},
    year = {2013},
    doi = {10.1007/978-3-662-03185-8}
}

@book{Kloeden1992,
  author    = {Peter E. Kloeden and Eckhard Platen},
  title     = {Numerical Solution of Stochastic Differential Equations},
  series    = {Applications of Mathematics},
  volume    = {23},
  publisher = {Springer},
  address   = {Berlin, Heidelberg},
  year      = {1992},
  doi       = {10.1007/978-3-662-12616-5}
}

@article{Cui2023,
  author  = {Tiangang Cui and Sergey Dolgov and Olivier Zahm},
  title   = {Scalable Conditional Deep Inverse Rosenblatt Transports Using Tensor Trains and Gradient-Based Dimension Reduction},
  journal = {Journal of Computational Physics},
  volume  = {485},
  pages   = {112103},
  year    = {2023},
  doi     = {10.1016/j.jcp.2023.112103}
}

@article{Katsoulakis2017,
  author  = {Markos A. Katsoulakis and Luc Rey-Bellet and Jie Wang},
  title   = {Scalable Information Inequalities for Uncertainty Quantification},
  journal = {Journal of Computational Physics},
  volume  = {336},
  pages   = {513--545},
  year    = {2017},
  doi     = {10.1016/j.jcp.2017.02.020}
}

@book{Kutoyants2004,
  author    = {Yury A. Kutoyants},
  title     = {{Statistical Inference for Ergodic Diffusion Processes}},
  series     = {Springer Series in Statistics},
  publisher = {Springer},
  address   = {London},
  year      = {2004},
  isbn      = {9781852335730}
}

@article{Harmandaris2016,
    title = {Path-space variational inference for non-equilibrium coarse-grained systems},
    journal = {J. Comput. Phys.},
    volume = {314},
    pages = {355-383},
    year = {2016},
    issn = {0021-9991},
    doi = {https://doi.org/10.1016/j.jcp.2016.03.021},
    author = {Vagelis Harmandaris and Evangelia Kalligiannaki and Markos Katsoulakis and Petr Plecháč},
}

@article{Dupuis2016,
    author = {Dupuis, Paul and Katsoulakis, Markos A. and Pantazis, Yannis and Plech\'{a}\v{c}, Petr},
    title = {Path-Space Information Bounds for Uncertainty Quantification and Sensitivity Analysis of Stochastic Dynamics},
    journal = {SIAM/ASA J. Uncertain. Quantif.},
    volume = {4},
    number = {1},
    pages = {80-111},
    year = {2016},
    doi = {10.1137/15M1025645},
}

@article{Tsourtis2015,
    author = {Anastasios Tsourtis and Yannis Pantazis and Markos Katsoulakis and Vagelis Harmandaris},
   doi = {10.1063/1.4922924},
   issn = {00219606},
   issue = {1},
   journal = {J. Chem. Phys.},
   month = {7},
   publisher = {American Institute of Physics Inc.},
   title = {Parametric sensitivity analysis for stochastic molecular systems using information theoretic metrics},
   volume = {143},
   year = {2015},
}

@article{Hartmann2012,
   author = {Carsten Hartmann and Christof Schütte},
   doi = {10.1088/1742-5468/2012/11/P11004},
   issn = {17425468},
   issue = {11},
   journal = {J. Stat. Mech. Theory Exp.},
   month = {11},
   title = {Efficient rare event simulation by optimal nonequilibrium forcing},
   volume = {2012},
   year = {2012},
}

@article{Zhang2021,
title = {Error bounds of the invariant statistics in machine learning of ergodic {Itô} diffusions},
journal = {Phys. D},
volume = {427},
pages = {133022},
year = {2021},
issn = {0167-2789},
doi = {https://doi.org/10.1016/j.physd.2021.133022},
author = {He Zhang and John Harlim and Xiantao Li},
}

@article{Opper2017,
   author = {Manfred Opper},
   doi = {10.1016/j.jcp.2016.11.021},
   issn = {10902716},
   journal = {J. Comput. Phys.},
   month = {2},
   pages = {127-133},
   publisher = {Academic Press Inc.},
   title = {An estimator for the relative entropy rate of path measures for stochastic differential equations},
   volume = {330},
   year = {2017},
}

@article{Mohamed2020,
author = {Mohamed, S. and Rosca, M. and Figurnov, M. and Mnih, A.},
journal = {J. Mach. Learn. Res.},
pages = {1--62},
title = {{Monte Carlo Gradient Estimation in Machine Learning}},
volume = {21},
year = {2020},
doi={10.5555/3455716.3455848}
}

@article{Lie2021,
archivePrefix = {arXiv},
eprint = {2106.09149},
author = {Lie, H. C.},
journal = {arXiv},
title = {{Fr\'echet derivatives of expected functionals of solutions to stochastic differential equations}},
year = {2021},
}

@article{Branicki2021,
  title     = {Lagrangian Uncertainty Quantification and Information Inequalities for Stochastic Flows},
  author    = {Branicki, Michal and Uda, Kenneth},
  journal   = {SIAM/ASA J. Uncertain. Quantif.},
  volume    = {9},
  number    = {3},
  pages     = {1242--1313},
  year      = {2021},
  doi       = {10.1137/19M1263133},
}

@article{Branicki2023,
  title     = {{Path-based divergence rates and Lagrangian uncertainty in stochastic flows}},
  author    = {Branicki, Michal and Uda, Kenneth},
  journal   = {SIAM J. Appl. Dyn. Syst.},
  volume    = {22},
  number    = {1},
  pages     = {1--72},
  year      = {2023},
  doi       = {10.1137/22M1491309},
}

@book{Karatzas1991,
  title={Brownian Motion and Stochastic Calculus},
  author={Karatzas, Ioannis and Shreve, Steven E.},
  series={Graduate Texts in Mathematics},
  volume={113},
  edition={2},
  year={1991},
  publisher={Springer},
  address={New York},
  doi={10.1007/978-1-4612-0949-2}
}

@article{Stuart2010,
   author = {Stuart, A. M.},
   title = {{Inverse problems: a Bayesian perspective}},
   journal={Acta Numerica},
   volume={19},
   pages={451-559},
   year = {2010},
   doi = {10.1017/S0962492910000061}
}

@book{Luenberger1969,
  author       = {Luenberger, David G.},
  title        = {Optimization by Vector Space Methods},
  publisher    = {John Wiley \& Sons, Inc.},
  address      = {New York},
  series       = {Series in Decision and Control},
  year         = {1969},
  pages        = {xvii+326},
  doi           = {10.1137/1012072}
}

@article{Lyubartsev1995,
  title = {{Calculation of effective interaction potentials from radial distribution functions: A reverse Monte Carlo approach}},
  author = {Lyubartsev, Alexander P. and Laaksonen, Aatto},
  journal = {Phys. Rev. E},
  volume = {52},
  issue = {4},
  pages = {3730--3737},
  year = {1995},
  publisher = {American Physical Society},
  doi = {10.1103/PhysRevE.52.3730}
}

@article{Bartok2010,
  title        = {Gaussian Approximation Potentials: The Accuracy of Quantum Mechanics, without the Electrons},
  author       = {Bartók, Albert P. and Payne, Mike C. and Kondor, Risi and Csányi, Gábor},
  journal      = {Physical Review Letters},
  volume       = {104},
  number       = {13},
  pages        = {136403},
  year         = {2010},
  doi          = {10.1103/PhysRevLett.104.136403},
}

@article{Lyubartsev2002,
  title        = {{On coarse-graining by the inverse Monte Carlo method: Dissipative particle dynamics simulations made to a precise tool in soft matter modeling}},
  author       = {Lyubartsev, A.\,P. and Karttunen, M. and Vattulainen, I. and Laaksonen, A.},
  journal      = {Soft Materials},
  volume       = {1},
  number       = {1},
  pages        = {121--137},
  year         = {2002},
  doi          = {10.1081/SMTS-120016746}
}

@article{Muskulus2011,
  author       = {Muskulus, Michael and Verduyn‐Lunel, Sjoerd},
  title        = {Wasserstein distances in the analysis of time series and dynamical systems},
  journal      = {Phys. D},
  volume       = {240},
  number       = {1},
  pages        = {45--58},
  year         = {2011},
  doi          = {10.1016/j.physd.2010.08.005},
}

@article{Yang2023,
author = {Yang, Yunan and Nurbekyan, Levon and Negrini, Elisa and Martin, Robert and Pasha, Mirjeta},
title = {Optimal Transport for Parameter Identification of Chaotic Dynamics via Invariant Measures},
journal = {SIAM J. Appl. Dyn. Syst.},
volume = {22},
number = {1},
pages = {269-310},
year = {2023},
doi = {10.1137/21M1421337},
}

@article{Botvinick2024,
  title = {Invariant Measures in Time-Delay Coordinates for Unique Dynamical System Identification},
  author = {Botvinick-Greenhouse, Jonah and Martin, Robert and Yang, Yunan},
  journal = {Phys. Rev. Lett.},
  volume = {135},
  issue = {16},
  pages = {167202},
  numpages = {7},
  year = {2025},
  publisher = {American Physical Society},
  doi = {10.1103/ppys-lx68},
}

@article{Song2021mle,
author = {Song, Yang and Durkan, Conor and Murray, Iain and Ermon, Stefano},
title = {Maximum likelihood training of score-based diffusion models},
year = {2021},
number = {109},
numpages = {14},
journal = {Proc. Mach. Learn. Res.},
url = {https://proceedings.mlr.press/v162/lu22f/lu22f.pdf}
}

@book{RevuzYor1999,
  author    = {Daniel Revuz and Marc Yor},
  title     = {Continuous Martingales and Brownian Motion},
  edition   = {3rd},
  year      = {1999},
  publisher = {Springer},
  volume    = {293},
  doi       = {10.1007/978-3-662-06400-9}
}

@book{Lelievre2010,
    author = {Lelièvre, Tony and Rousset, Mathias and Stoltz, Gabriel},
    title = {{Free Energy Computations}},
    publisher = {Imperial College Press},
    year = {2010},
    doi = {10.1142/p579},
}

@article{Frohlking2024,
    author = {Fröhlking, Thorben and Bonati, Luigi and Rizzi, Valerio and Gervasio, Francesco Luigi},
    title = {Deep learning path-like collective variable for enhanced sampling molecular dynamics},
    journal = {J. Chem. Phys.},
    volume = {160},
    number = {17},
    pages = {174109},
    year = {2024},
    month = {05},
    doi = {10.1063/5.0202156},
}

@article{MullerBrown1979,
  author    = {Klaus M{\"u}ller and Leo D. Brown},
  title     = {Location of saddle points and minimum energy paths by a constrained simplex optimization procedure},
  journal   = {Theoretica Chimica Acta},
  year      = {1979},
  volume    = {53},
  number    = {1},
  pages     = {75--93},
  doi       = {10.1007/BF00547608},
}

@article{Rainforth2018,
  title     = {Tighter Variational Bounds are Not Necessarily Better},
  author    = {Rainforth, Tom and Kosiorek, Adam R. and Le, Tuan Anh and Maddison, Chris J. and Igl, Maximilian and Wood, Frank and Teh, Yee Whye},
  journal    = {Proc. Mach. Learn. Res.},
  volume    = {80},
  pages     = {4277--4285},
  year      = {2018},
  url       = {https://proceedings.mlr.press/v80/rainforth18b/rainforth18b.pdf}
  
}

@article{Duchi2011,
  title={Adaptive subgradient methods for online learning and stochastic optimization},
  author={Duchi, John and Hazan, Elad and Singer, Yoram},
  journal={J. Mach. Learn. Res.},
  volume={12},
  pages={2121--2159},
  year={2011},
  doi={10.5555/1953048.2021068}
}

@article{Batzner2022,
   author = {Simon Batzner and Albert Musaelian and Lixin Sun and Mario Geiger and Jonathan P. Mailoa and Mordechai Kornbluth and Nicola Molinari and Tess E. Smidt and Boris Kozinsky},
   journal = {Nature Communications},
   volume = {13},
   number = {2453},
   title = {E(3)-equivariant graph neural networks for data-efficient and accurate interatomic potentials},
   year = {2022},
   doi = {10.1038/s41467-022-29939-5}
}

@article{Musaelian2023,
   author = {Musaelian, A. and Batzner, S. and Johansson, A. and et al.},
   journal = {Nature Communications},
   volume = {14},
   number = {579},
   title = {Learning local equivariant representations for large-scale atomistic dynamics},
   year = {2023},
   doi = {10.1038/s41467-023-36329-y}
}

@article{Allen2009,
  author    = {Rosalind J. Allen and Chantal Valeriani and Pieter Rein ten Wolde},
  title     = {Forward flux sampling for rare event simulations},
  journal   = {Journal of Physics: Condensed Matter},
  volume    = {21},
  number    = {46},
  pages     = {463102},
  year      = {2009},
  doi       = {10.1088/0953-8984/21/46/463102}
}

@incollection{Voter2005,
  author    = {Arthur F. Voter},
  title     = {Introduction to the kinetic {Monte Carlo} Method},
  booktitle = {Radiation Effects in Solids},
  series    = {NATO Science Series},
  volume    = {235},
  pages     = {1--23},
  year      = {2005},
    publisher = {Springer},
  doi       = {10.1007/978-1-4020-5295-8_1}
}

@article{Dellago1998,
  author    = {Christoph Dellago and Pieter G. Bolhuis and Ferenc S. Csajka and David Chandler},
  title     = {Transition path sampling and the calculation of rate constants},
  journal   = {J. Chem. Phys.},
  volume    = {108},
  pages     = {1964--1977},
  year      = {1998},
  doi       = {10.1063/1.475562}
}

@article{Feng2021,
  author    = {Chunrong Feng and Huaizhong Zhao and Johnny Zhong},
  title     = {Expected Exit Time for Time‑Periodic Stochastic Differential Equations and Applications to Stochastic Resonance},
  journal   = {Phys. D},
  volume    = {417},
  pages     = {132815},
  year      = {2021},
  doi       = {10.1016/j.physd.2020.132815}
}

@article{Drautz2019,
  title = {Atomic cluster expansion for accurate and transferable interatomic potentials},
  author = {Drautz, Ralf},
  journal = {Phys. Rev. B},
  volume = {99},
  issue = {1},
  pages = {014104},
  numpages = {15},
  year = {2019},
  publisher = {American Physical Society},
}

@article{Vempala2019,
    author = {Santosh S. Vempala and Andre Wibisono},
    title = {{Rapid convergence of the unadjusted Langevin algorithm: isoperimetry suffices}},
    journal = {Adv. Neural Inf. Process. Syst.},
    year = {2019},
    doi = {10.1007/978-3-031-26300-2_15}
}

@article{Roberts1996,
    author = {Gareth O. Roberts and Richard L. Tweedie},
    title = {{Exponential convergence of Langevin distributions and their discrete approximations}},
    journal = {Bernoulli},
    volume = {2},
    number = {4},
    year = {1996},
    pages = {341-363}
}

@article{Dalalyan2017,
  title   = {Theoretical guarantees for approximate sampling from smooth and log‑concave densities},
  author  = {Dalalyan, Arnak S.},
  journal = {J. R. Stat. Soc. Ser. B. Stat. Methodol.},
  volume  = {79},
  number  = {3},
  pages   = {651--676},
  year    = {2017},
  doi     = {10.1111/rssb.12183}
}

@article{Zhang2016,
  author    = {Xicheng Zhang},
  title     = {Stochastic differential equations with {Sobolev} diffusion and singular drift and applications},
  journal   = {Ann. Appl. Probab.},
  volume    = {26},
  number    = {5},
  pages     = {2697--2732},
  year      = {2016},
  publisher = {Institute of Mathematical Statistics},
  doi       = {10.1214/15-AAP1140}
}

@article{Shao2018,
  title     = {{Weak convergence of Euler–Maruyama’s approximation for stochastic differential equations under integrability condition}},
  author    = {Shao, Jinghai},
  journal   = {arXiv},
  archivePrefix = {arXiv},
  eprint    = {1808.07250},
  year      = {2018},
}

@article{Yu2024,
  author  = {Yu, Yuan},
  title   = {{Convergence of relative entropy for Euler–Maruyama scheme to stochastic differential equations with additive noise}},
  journal = {Entropy},
  volume  = {26},
  number  = {3},
  pages   = {232},
  year    = {2024},
  doi     = {10.3390/e26030232}
}

@article{Henin2022,
    title={{Enhanced Sampling Methods for Molecular Dynamics Simulations}},
    volume={4},
    number={1},
    journal={Living J. Comp. Mol. Sci.},
    author={Hénin, Jérôme and Lelièvre, Tony and Shirts, Michael R. and Valsson, Omar and Delemotte, Lucie},
    year={2022},
    pages={1583},
    doi={10.33011/livecoms.4.1.1583},
}

@article{Yang2019,
    author = {Yang, Yi Isaac and Shao, Qiang and Zhang, Jun and Yang, Lijiang and Gao, Yi Qin},
    title = {Enhanced sampling in molecular dynamics},
    journal = {J. Chem. Phys.},
    volume = {151},
    number = {7},
    pages = {070902},
    year = {2019},
    month = {08},
    issn = {0021-9606},
    doi = {10.1063/1.5109531},
}

@article{Mori2020,
    author = {Mori, Yusuke and Okazaki, Kei-ichi and Mori, Toshifumi and Kim, Kang and Matubayasi, Nobuyuki},
    title = {Learning reaction coordinates via cross-entropy minimization: Application to alanine dipeptide},
    journal = {J. Chem. Phys.},
    volume = {153},
    number = {5},
    pages = {054115},
    year = {2020},
    month = {08},
    issn = {0021-9606},
    doi = {10.1063/5.0009066},
}

@article{Jain2013,
  title={{Commentary: The Materials Project: A materials genome approach to accelerating materials innovation}},
  author={Jain, Anubhav and Ong, Shyue Ping and Hautier, Geoffroy and Chen, Weike and Richards, William Davidson and Dacek, Stephen and others},
  journal={APL Materials},
  volume={1},
  number={1},
  pages={011002},
  year={2013},
  publisher={AIP Publishing},
  doi={10.1063/1.4812323},
}

@article{Ramakrishnan2014,
  title     = {Quantum chemistry structures and properties of 134 kilo molecules},
  author    = {Ramakrishnan, Raghunathan and Dral, Pavlo O. and Rupp, Matthias and von Lilienfeld, O. Anatole},
  journal   = {Scientific Data},
  volume    = {1},
  pages     = {140022},
  year      = {2014},
  publisher = {Nature Publishing Group},
  doi       = {10.1038/sdata.2014.22}
}

@article{Lecuyer1995,
  title     = {Note: On the Interchange of Derivative and Expectation for Likelihood Ratio Derivative Estimators},
  author    = {L'Ecuyer, Pierre},
  journal   = {Management Science},
  volume    = {41},
  number    = {4},
  pages     = {738--747},
  year      = {1995},
  publisher = {INFORMS},
  doi       = {10.1287/mnsc.41.4.738}
}

@Inbook{Dashti2017,
author={Dashti, Masoumeh and Stuart, Andrew M.},
title={{The Bayesian Approach to Inverse Problems}},
bookTitle={Handbook of Uncertainty Quantification},
year={2017},
publisher={Springer International Publishing},
pages={311--428},
doi={10.1007/978-3-319-12385-1_7},
}
\end{document}